%% file: paper.tex
\documentclass[acmsmall,nonacm]{acmart}
\AtBeginDocument{%
  }
\allowdisplaybreaks
\usepackage{blindtext}
\usepackage[ruled,vlined,linesnumbered]{algorithm2e}
\usepackage{graphicx,enumitem,booktabs,amsmath}
\usepackage{psfrag}
\usepackage{marginnote}
\usepackage{tikz}
\usepackage{listings}
\usepackage[tableposition=top]{caption}
\usepackage{subfigure}
\usepackage{color}
\usepackage{hyperref}

\usepackage{mathabx}
\usepackage[utf8]{inputenc}
\usepackage{xcolor,colortbl}
\usepackage{hhline}
\usepackage{cancel}
\usepackage{tikzscale}
\usepackage{soul} 
\usepackage{balance}
\usepackage{multirow}
\renewcommand{\paragraph}[1]{\smallskip \noindent{\bf {#1. }}}
\newcommand{\randint}{\texttt{RandInt}}
\newcommand{\OUT}{\textsf{\upshape OUT}}
\newcommand{\true}{\textup{\textbf{true}}}
\newcommand{\false}{\textup{\textbf{false}}}
\newcommand{\D}{\mathcal{D}}
\newcommand{\E}{\mathcal{E}}
\newcommand{\Q}{\mathcal{Q}}
\newcommand{\V}{\mathcal{V}}

\newcommand{\dom}{\mathrm{dom}}  
\newcommand{\y}{\mathbf{y}}
\newcommand{\R}{\mathcal{R}}

\renewcommand{\O}{\Tilde{O}}
\newcommand{\cut}[1]{{}}
\newcommand{\adom}{\mathsf{adom}}
\newcommand{\heavy}{\mathsf{heavy}}
\newcommand{\light}{\mathsf{light}}
\newcommand{\expt}{\mathbb{E}}
\newcommand{\var}{\textsf{Var}}
\renewcommand{\star}{\Q_{\textsf{star}}}
\newcommand{\chain}{\Q_{\textsf{chain}}}
\renewcommand{\matrix}{\Q_{\textup{\textsf{matrix}}}}

\input{def/jc-def}

\begin{document}

    \author{Xiao Hu}
	\email{xiaohu@uwaterloo.ca}
	\orcid{0000-0002-7890-665X}
	\affiliation{
		\institution{University of Waterloo}
		\streetaddress{200 University Ave W}
		\city{Waterloo}
		\state{Ontario}
		\country{Canada}
		\postcode{N2L 3G1}
	}
    
    \author{Jinchao Huang}
	\email{jchuang@se.cuhk.edu.hk}
	\orcid{0009-0009-2902-259X}
	\affiliation{
		\institution{The Chinese University of Hong Kong}
		\streetaddress{}
		\city{Shatin}
		\country{Hong Kong SAR}
		\postcode{}
	}
    \renewcommand{\shortauthors}{X. Hu and J. Huang}

    \title{Towards Output-Optimal Uniform Sampling and Approximate Counting for Join-Project Queries}
    
    \begin{abstract}
    Uniform sampling and approximate counting are fundamental primitives for modern database applications, ranging from query optimization to approximate query processing. While recent breakthroughs have established optimal sampling and counting algorithms for full join queries, a significant gap remains for join-project queries, which are ubiquitous in real-world workloads. The state-of-the-art ``propose-and-verify'' framework \cite{chen2020random} for these queries suffers from fundamental inefficiencies, often yielding prohibitive complexity when projections significantly reduce the output size.

    In this paper, we present the first asymptotically optimal algorithms for fundamental classes of join-project queries, including matrix, star, and chain queries. By leveraging a novel rejection-based sampling strategy and a hybrid counting reduction, we achieve polynomial speedups over the state of the art. We establish the optimality of our results through matching communication complexity lower bounds, which hold even against algebraic techniques like fast matrix multiplication. Finally, we delineate the theoretical limits of the problem space. While matrix and star queries admit efficient sublinear-time algorithms, 
    we establish a significantly stronger lower bound for chain queries, demonstrating that sublinear algorithms are impossible in general.
    \end{abstract}
    
	
	
	

    \maketitle
   
    \section{Introduction}
    \label{sec:intro}
    Joins constitute the backbone of relational database systems, serving as the fundamental primitive for combining data across multiple relations. However, in the era of big data, the sheer volume of join results often renders exact evaluation computationally prohibitive. While worst-case optimal join algorithms~\cite{ngo2012worst, ngo2018worst} have settled the worst-case complexity up to the AGM bound~\cite{atserias2008size}, the join size can still be polynomially larger than the input size $N$ (i.e., the number of tuples in the database), reaching up to $N^{\rho^*}$, where $\rho^*$ is the fractional edge covering number of the underlying join query. Consequently, for many downstream applications---such as approximate query processing~\cite{agarwal2013blinkdb}---computing the exact result is both expensive and unnecessary. This computational bottleneck has necessitated a shift towards \emph{uniform sampling} and \emph{approximate counting}, which share strong algorithmic connections and complexity characteristics. The goal is to generate random samples from the query result or estimate the output size without materializing the full query result.
    
    For full join queries (where all attributes are retained without projection), these two problems have been extensively studied by the database community~\cite{olken1986simple,olken1995random,acharya1999join,chaudhuri1999random,zhao2018random,deng2023join,kim2023guaranteeing} and the algorithm community \cite{eden2017approximately, assadi2019simple}.  \citet{acharya1999join} and \citet{chaudhuri1999random} initiated the investigation of uniform sampling over joins in 1999, identifying the primary barrier: the sampling operator cannot be pushed down, i.e., $\textsf{sample}(R \Join S) \neq \textsf{sample}(R) \Join \textsf{sample}(S)$. To overcome this, research focused on precomputing indices to guide sampling. \citet{zhao2018random} demonstrated that for acyclic joins, an index can be built in $O(N)$ time to support $O(1)$ uniform sampling. More recently, \citet{kim2023guaranteeing} and \citet{deng2023join} tackled general joins (including cyclic joins) by showing that an index can be built in $O(N)$ time to support $\O\left(\frac{N^{\rho^*}}{\OUT_\Join}\right)$ uniform sampling, where $\OUT_\Join$ is the size of the full join result. Parallel to uniform sampling, \citet{kim2023guaranteeing} showed that the same algorithmic principles apply to approximate counting. The size of the full join, $\OUT_\Join$, can be estimated with the same complexity once the input is preprocessed to support some basic primitives (see Section~\ref{sec:problem-definition}). Moreover, for acyclic joins, this bound can be further improved to $\O\left(\min\{N, \frac{N^{\rho^*}}{\OUT_\Join}\}\right)$ by leveraging the exact counting algorithm~\cite{yannakakis1981algorithms}.
    
    However, a significant gap remains in the literature: real-world analytical queries rarely retain all attributes; they frequently involve projections. A join-project query asks for distinct tuples projected onto a subset of attributes. Unfortunately, the optimal techniques developed for full joins do not generalize to arbitrary projections. The state-of-the-art sampling framework for join-project queries~\cite{chen2020random} relies on a ``propose-and-verify'' approach: it proposes a candidate tuple from the \emph{projection-induced join} (the join query induced by projecting input relations onto the output attributes) and subsequently verifies if this candidate extends to a valid answer in the full join. This approach faces two fundamental bottlenecks: (1) \emph{low acceptance probability}: the size of the projection-induced join space is often significantly larger than the actual projected result size $\OUT$; and (2) \emph{high verification cost}: verifying whether a candidate tuple extends to a full join result is computationally expensive, often as costly as evaluating the residual join query induced by non-output attributes. Consequently, this framework yields a sampling complexity of $\O\left( \frac{N^{\rho^*}}{\OUT}\right)$. This also remains the state-of-the-art for approximate counting via a reduction to uniform sampling~\cite{chen2020random}.
    
    In this paper, we break the stalemate for join-project queries by developing the first asymptotically output-optimal algorithms for fundamental query classes, including matrix, star, and chain queries (formally defined in Section~\ref{sec:problem-definition}). 
    Crucially, our lower bounds are derived from communication complexity, which applies to {\em any} algorithm, and therefore effectively rule out the possibility of circumvention by algebraic techniques such as fast matrix multiplication. Furthermore, our hardness result for chain queries rules out the possibility of sublinear algorithms for general join-project queries. These results collectively mark a significant step towards characterizing the fine-grained complexity of uniform sampling and approximate counting for general join-project queries.
        
    \subsection{Problem Definitions}
    \label{sec:problem-definition}
    \noindent {\bf Join-Project Query.} A join-project query can be modeled as a triple $\Q = (\V, \mathcal{E}, \y)$, 
    where $\V$ is a set of attributes, $\mathcal{E} = \{e_1, e_2, \ldots, e_k\} \subseteq 2^\V$ is a set of schemas, and $\y \subseteq \bigcup_{e\in \E} e$ is the set of output attributes. For simplicity, we assume $\V = \bigcup_{e\in \E}e$. Let $\dom(A)$ be the domain of an attribute $A \in \V$. Let $\dom(X) = \prod_{A \in X} \dom(A)$ be the domain of a subset of attributes $X \subseteq \V$.
    A \emph{tuple} $t$ defined over a subset of attributes $X \subseteq \V$ is a function that assigns a value from $\dom(A)$ to every attribute $A \in X$. 
    An \emph{instance} of $\Q$ is a set of relations $\R = \{R_e : e \in \mathcal{E}\}$, where each relation $R_e$ is a set of {\em tuples} defined over attributes $e$. 
    For a tuple defined over attributes $X$ and any subset of attributes $Y \subseteq X$,  we denote by $\pi_Y t$ the projection of $t$ onto $Y$. The query result of $\Q$ on $\R$ is defined as:
    \[
        \Q(\R) = \{t \in \dom(\y) \mid \exists t' \in \dom(\V), \pi_\y t' = t, \forall e \in \mathcal{E}, \pi_e t' \in R_e \}.
    \]
    \noindent We use $N = \sum_{e \in \E} |R_e|$ to denote the {\em input size}, i.e., the total number of tuples in input relations, and $\OUT = |\Q(\R)|$ to denote the {\em output size}, i.e., the number of tuples in the query result. 
    
    In plain language, a join-project query first computes the full join result of the underlying relations in $\Q$, which is all combinations of tuples, one from each relation, that share the same value on all common attributes, and then computes the projection (without duplication) of the full join result onto the output attributes. Join-project queries include many commonly seen database queries as special cases. For example, if we take $\V = \{A,B,C\}$ with $\y = \{A,C\}$, and $\E= \{\{A,B\}, \{B,C\}\}$, it becomes the matrix query (denoted as $\matrix$); if we take $\V = \{A_1,A_2,\cdots, A_k,B\}$ with $\y = \{A_1,A_2,\cdots,A_k\}$, and $\E= \{\{A_1,B\}, \{A_2,B\}, \cdots, \{A_k,B\}\}$, it becomes the $k$-star query (denoted as $\star$); and if we take $\V = \{A_1,A_2,\cdots, A_{k+1}\}$ with $\y = \{A_1,A_{k+1}\}$, and $\E= \{\{A_1,A_2\}, \{A_2,A_3\}, \cdots, \{A_k,A_{k+1}\}\}$, it becomes the $k$-chain query (denoted as $\chain$). 

    An important operator used together with join is semi-join. For two relations $R_e$ and $R_{e'}$, the semi-join $R_e \ltimes R_e'$ finds the set of tuples in $R_e$ that can be joined with at least one tuple from $R_{e'}$. For some tuple $t \in R_{e'}$, the semi-join $R_e\ltimes t$ finds the set of tuples in $R_e$ that can be joined with $t$.

    \paragraph{Uniform Sampling and Approximate Counting for Join-project Queries} In this paper, we are mainly interested in the following two problems. The first concerns building a data structure for generating random samples from the query result:

    \begin{problem}[Uniform Sampling over Join-Project Query]
    \label{prob:static-index}
     Given a join-project query $\Q$ and an instance $\R$, 
     the goal is to build a data structure in linear time that supports efficient generation of uniform samples from the query result $\Q(\R)$.
     Additionally, the samples returned to distinct requests must be independent. 
    \end{problem}

    \noindent The second problem concerns estimating the size of the query result. To quantify the accuracy of our estimation, we first define the standard notion of multiplicative approximation:
    For a non-negative value $V$ and a parameter $0 < \epsilon < 1$, $\widehat{V}$ is an \emph{$\epsilon$-approximation} of $V$ if $(1-\epsilon)V \le \widehat{V} \le (1+\epsilon)V.$

    \begin{problem}[Approximate Counting for Join-Project Query]
     Given parameters $0<\epsilon,\delta <1$, a join-project query $\Q$ and an instance $\R$, the goal is to return an $\epsilon$-approximation $\widehat{\OUT}$ of the output size $\OUT$ with probability at least $1-\delta$.
    \end{problem}
    For the uniform sampling problem, we are interested in the {\em sampling time} of generating a uniform sample from the preprocessed data structure.
    For the approximate counting problem, we are interested in the total runtime. For both problems, we study the data complexity by considering the query size (i.e., $k$) as constants and measuring the complexity of algorithms by data-dependent quantities, such as input size $N$ and output size $\OUT$, and additional parameters, such as approximation quality $\epsilon$ and the failure probability $\delta$.

    \paragraph{Model of Computation} We use the standard word RAM model with uniform cost measures \cite{cormen2022introduction}. For an instance of size $N$, the machine operates on words of length $w = \Omega(\log N)$ bits. 
    Accessing a memory location takes $O(1)$ time. Performing basic arithmetical operations (such as addition, multiplication, and division) on two words takes $O(1)$ time. 
    We assume the machine can generate a random integer uniformly from a range $[a, b]$ in $O(1)$ time, provided the range fits within a word. We denote this operation as $\texttt{RandInt}(a, b)$.
    
    In addition, for both problems on graphs (as a special case of self-joins), 
    researchers have adopted the {\em property testing model}~\cite{assadi2019simple}, which assumes that the graph has already been stored in some standard graph data structure (e.g., adjacency lists with hash tables), so that some basic primitives are supported in $O(1)$ time. This model was subsequently extended to the relational setting \cite{kim2023guaranteeing}, so that the following primitives are supported in $\O(1)$ time:
    \begin{itemize}[leftmargin=*]
        \item \textbf{(sample tuple)} For a relation $R_e \in \R$, it returns a uniform sample from $R_e$. 
        \item \textbf{(compute degree)} For a relation $R_e \in \R$, a tuple $t$ defined on attributes $e' \subseteq e$, it returns $|R_e \ltimes t|$, i.e., the number of tuples from $R_e$ that match $t$ on attributes $e'$.
        \item \textbf{(access neighbor)} For a relation $R_e \in \R$, a tuple $t$ defined on attributes $e' \subseteq e$, and an integer $j$, it returns the $j$-th tuple in $\pi_{e-e'} (R_e \ltimes t)$ under a fixed ordering.\footnote{The primitives in \cite{kim2023guaranteeing} are more powerful with respect to {\em access neighbor}: for a relation $R_e$, a tuple $t$ defined on attributes $e' \subseteq e$, a subset $I \subseteq e-e'$, and an integer $j$, it returns the $j$-th tuple in $\pi_I(R_e \ltimes t)$ under a fixed ordering. 
        } 
        
        \item \textbf{(test tuple)} For a relation $R_e \in \R$ and a tuple $t$, it returns true if $t \in R_e$ and false otherwise. 
    \end{itemize}

      \noindent {\bf Other Notations and Assumptions.} We denote $q = (\V,\E)$ as the full join of the join-project query $\Q = (\V, \mathcal{E}, \y)$. 
      Given a join query $q=(\V,\E)$, a {\em fractional edge covering} is a function $w$ that assigns a non-negative weight $w(e)$ to each schema $e \in \E$ such that for every attribute $A \in \V$, the sum of weights of all schema containing $A$ is at least 1, i.e., $\sum_{e \in \E: A \in e} w(e) \ge 1$. The {\em fractional edge covering number} $\rho^*$ is the minimum possible total weight, $\sum_{e \in \E} w(e)$, over all fractional edge coverings $w$.  
      Note that $\rho^*=k$ for $k$-star join, and $\rho^* = \lceil \frac{k+1}{2}\rceil$ for $k$-chain join. 
      
      For an instance $\R$, we denote the size of the full join results of $q$ as $\OUT_\Join$.
      The \emph{active domain} of an attribute $A \in \V$, denoted as $\adom_\R(A)$, is the set of values from $\dom(A)$ appearing in some tuple of $\R$, i.e., $\adom_\R(A)=\{a\in \dom(A)\mid \exists e\in \E, \exists t\in R_e, \pi_A t=a\}$. When the context is clear, we use $\adom(A)$ to denote $\adom_\R(A)$. For any $n \in \mathbb{Z}^+$, we use $[n]$ to denote $\{1,2,\cdots,n\}$. Every logarithm has base 2 by default. For a pair of sets $S_1$ and $S_2$, we use $S_1- S_2 = \{x \in S_1: x \notin S_2\}$ to denote the set minus operation. We use $\widetilde{O}(\cdot)$ to suppress any poly-logarithmic factors. 
      
      \paragraph{Handling Empty Results} For clarity of presentation, we focus on the case when $\OUT \ge 1$. The edge case $\OUT=0$ is handled using a standard timeout strategy~\cite{chen2020random, kim2023guaranteeing}. Since the complexities of our algorithms are monotonically non-increasing with respect to $\OUT$ (e.g., $O(\frac{N}{\sqrt{\OUT}})$), the worst-case runtime upper bound occurs when $\OUT=1$. Consequently, the case of $\OUT=0$ can be detected by running the algorithm with a time budget corresponding to a constant times the complexity for $\OUT=1$. If the algorithm fails to generate a sample or return an estimate within this budget (repeated for $\O(1)$ trials to satisfy error probability $\delta$), we terminate and report $\OUT=0$ with high probability. This timeout cost 
      satisfies our general complexity bounds. 
      With a slight abuse of notation, we use terms like $O(\frac{N}{\sqrt{\OUT}})$ to denote $O(\frac{N}{\sqrt{\max\{1, \OUT\}}})$ in the paper.

    \subsection{Previous Results}
    \label{sec:previous}
    Recall that full joins are special cases of join-project queries where all attributes are output attributes. An index can be built in $O(N)$ time such that one uniform sample can be generated in $O(1)$ time for acyclic joins\footnote{A join query $q=(\V,\E)$ is $\alpha$-acyclic \cite{yannakakis1981algorithms} if there exists a join tree such that (i) there is a one-to-one correspondence between tree nodes and relations; and (ii) for each attribute, the set of nodes containing it forms a connected subtree.} \cite{zhao2018random} and in
    $\O(\frac{N^{\rho^*}}{\OUT})$  
    time for cyclic joins \cite{deng2023join, kim2023guaranteeing}\footnote{
    \cite{kim2023guaranteeing} does not have a preprocessing cost in the (stronger) property testing model. However, preprocessing the extra primitive they introduce still incurs an $O(N)$ cost. For fair comparison, all results remain within the standard property testing model.
    }. The previous result for acyclic joins can be extended to the class of free-connex queries\footnote{A join-project query $\mathcal{Q}= (\V,\E,\y)$ is free-connex \cite{bagan2007acyclic} if $(\V,\E)$ is acyclic and $(\V, \E \cup \{\y\})$ (i.e., by adding another relation containing all output attributes) is also acyclic.}, since a free-connex query can be transformed into an acyclic full join in $O(N)$ time. However, this is a very limited class of join-project queries; for example, matrix, star, and chain queries are not free-connex.

    \citet{chen2020random} first investigated this problem for join-project queries with arbitrary projections. They proposed the ``propose-and-verify'' framework: first, propose a uniform sample from the \emph{projection-induced full join}; second, accept the candidate if it can be extended to a full join result of $q$. By integrating the state-of-the-art algorithms \cite{deng2023join, kim2023guaranteeing} for sampling a full join result, an index can be built in $O(N)$ time, such that one uniform sample can be generated in $\O(\frac{N^{\rho^*}}{\OUT})$ time, where $\rho^*$ is the fractional edge covering number of the underlying full join $q$. Hence, the best sampling time is $\O(\frac{N^2}{\OUT})$ for matrix queries, $\O(\frac{N^k}{\OUT})$ for $k$-star queries, and $\O(\frac{N^{\lceil \frac{k+1}{2}\rceil}}{\OUT})$ for $k$-chain queries. 
    
    \citet{chen2020random} also showed a reduction from approximate counting to uniform sampling for general join-project queries. 
    If an index can be built in $O(N)$ time and returns a uniform sample from the query results in $O(1)$ time with probability $\frac{\OUT}{N^{\rho^*}}$, then a constant-approximation of the output size can be returned with constant probability in $\O(N + \frac{N^{\rho^*}}{\OUT})$ time.
    
    Additionally, approximate counting and (near) uniform sampling\footnote{Unlike perfect uniform sampling algorithms, a near-uniform sampling algorithm may not return each query result with strictly same probability. The sampling algorithm we derive for chain queries is also near-uniform.} were previously studied by \cite{arenas2021approximate} for general join-project queries. Their work established that a fully polynomial-time randomized approximation scheme exists if and only if the query has bounded hypertree width. In contrast, our work focuses on the fine-grained complexity of these problems, specifically characterizing the runtime in terms of both the input size $N$ and the output size $\OUT$.

     \begin{figure*}[t]
  {\small
  \renewcommand{\arraystretch}{1.4} 
  \setlength{\tabcolsep}{3pt}
  \begin{tabular}{c|c|c|c|c|c}
    \toprule
     \multicolumn{2}{c|}{\multirow{2}{*}{Problem}} & \multirow{2}{*}{\textbf{Matrix}} & \multirow{2}{*}{\textbf{Star}} & \multicolumn{2}{c}{\textbf{Chain}} \\
    \cline{5-6}
    \multicolumn{2}{c|}{} &  & & $k=3$ & $k\ge 4$\\
    \hline
    \hline
    \multirow{5}{*}{\rotatebox[origin=c]{90}{\shortstack[c]{\textbf{Uniform}\\\textbf{Sampling}}}}
      & Prior Work~\cite{chen2020random, kim2023guaranteeing}& 
        $O(\frac{N^2}{\OUT})$ & 
        $O(\frac{N^k}{\OUT})$ & 
        $O(\frac{N^{2}}{\OUT})$ & $O(\frac{N^{\lceil \frac{k+1}{2}\rceil}}{\OUT})$\\
      \cline{2-6}
      & \multirow{2}{*}{Our UB} & 
        {\color{red} $O(\frac{N}{\sqrt{\OUT}})$} & 
        {\color{red} $O(\frac{N}{\OUT^{1/k}})$} & 
        {\color{red} $O(\min\{N, \frac{N^2}{\OUT}\})$} & {\color{red} $O(N)$}\\
      & & 
        {\color{red} Theorem~\ref{thm:sample-matrix-ub}} & 
        {\color{red} Theorem~\ref{thm:star-sample-ub}} & 
        {\color{red} Theorem \ref{thm:chain-sample-ub}, \cite{chen2020random}} & {\color{red} Theorem \ref{thm:chain-sample-ub}} \\
      \cline{2-6}
      & \multirow{2}{*}{Our LB} & 
        {\color{red} $\Omega(\frac{N}{\sqrt{\OUT}})$} & 
        {\color{red} $\Omega(\frac{N}{\OUT^{1/k}})$} & 
        {\color{red} $\Omega(\min\{N, \frac{N^2}{\OUT}\})$} & {\color{red} $\Omega(N)$}\\
      & & 
        {\color{red} Theorem~\ref{thm:sample-matrix-combinatorial-lb}} & 
        {\color{red} Theorem~\ref{thm:star-sample-lb}} & 
        {\color{red} Theorem~\ref{thm:chain-sample-lb-3}} & {\color{red} Theorem~\ref{thm:chain-sample-lb-4}} \\
    \hline
    \hline
    \multirow{5}{*}{\rotatebox[origin=c]{90}{\shortstack[c]{\textbf{Approximate}\\\textbf{Counting}}}} 
      & Prior Work~\cite{chen2020random, kim2023guaranteeing} & 
        $O(\frac{N^2}{\OUT})$ & 
        $O(\frac{N^k}{\OUT})$ & 
        $O(\frac{N^2}{\OUT})$ & $O( \frac{N^{\lceil \frac{k+1}{2}\rceil}}{\OUT})$ \\
      \cline{2-6}
      & \multirow{2}{*}{Our UB} & 
        {\color{red} $O(\frac{N}{\sqrt{\OUT}})$} & 
        {\color{red} $O(\frac{N}{\OUT^{1/k}})$} & 
        {\color{red} $O(\min\{N, \frac{N^2}{\OUT}\})$} & 
        {\color{red} $O(N)$}\\
      & & 
        {\color{red} Theorem~\ref{thm:count-matrix-up}} & 
        {\color{red} Theorem~\ref{thm:star-count-ub}} & 
        {\color{red} Theorem~\ref{thm:chain-count-ub}, \cite{chen2020random}} & 
        {\color{red} Theorem~\ref{thm:chain-count-ub}}\\
      \cline{2-6}
      & \multirow{2}{*}{Our LB} & 
        {\color{red} $\Omega(\frac{N}{\sqrt{\OUT}})$} & 
        {\color{red} $\Omega(\frac{N}{\OUT^{1/k}})$} & 
        {\color{red} $\Omega(\min\{N, \frac{N^2}{\OUT}\})$} & {\color{red} $\Omega(N)$} \\
      & & 
        {\color{red} Theorem~\ref{thm:count-matrix-lb}} & 
        {\color{red} Theorem~\ref{thm:star-count-lb}} & 
        {\color{red} Theorem~\ref{thm:count-chain-lb-3}} & 
        {\color{red} Theorem~\ref{thm:count-chain-lb-4}} \\
    \bottomrule
  \end{tabular}
  }
  \vspace{-1em}
  \caption{Comparison between prior results and our new results (in red) in the property testing model. $N$ is the input size, $\OUT$ is the output size, and $k$ is the number of relations. For uniform sampling, we assume $O(N)$ preprocessing time. For matrix and star queries, the uniform samplers always return every query result with the exact same probability, while for the chain query, the sampler is near-uniform (i.e., uniform with high probability). The lower bounds for uniform sampling over chain queries only apply to $W$-uniform sampling algorithms. For approximate counting, we assume $\epsilon$ is a small constant and $\delta = 1/N^{O(1)}$.}
  \label{fig:summary-main}
\end{figure*}
\footnotetext{A sampling algorithm is called {\em $W$-uniform sampling} if given an index built during preprocessing, a uniform sample from the query result can be drawn with probability $\frac{\OUT}{W}$ for some known value $W > \OUT$. Note that all sampling algorithms presented in this paper fall into this class, as do all prior sampling algorithms for join-project queries that we are aware of. }

   \subsection{Our Results} 
    \label{sec:our-results}
    Our main results are summarized in Figure~\ref{fig:summary-main}. All our indices built for uniform sampling use $O(N)$ preprocessing time. For simplicity, the lower bound shown in Figure~\ref{fig:summary-main} assumes $O(N)$ preprocessing time. Below, we focus on the sample generation time and approximate counting time. 

    \noindent
    {\bf Section~\ref{sec:sample-matrix} (uniform sampling over matrix query).} We begin with uniform sampling over matrix queries. We demonstrate that a uniform sample can be generated in $O\left(\frac{N}{\sqrt{\OUT}}\right)$ expected time. This strictly outperforms the prior best bound of $O\left(\frac{N^2}{\OUT}\right)$, yielding an improvement factor of $O\left(\frac{N}{\sqrt{\OUT}}\right)$. Since $\OUT$ ranges from $1$ to $N^2$, this implies a polynomial speedup when $\OUT = o(N^2)$.

    \noindent
    {\bf Section~\ref{sec:count-matrix} (approximate counting for matrix query).} We next address approximate counting. We introduce a novel algorithm based on a hybrid reduction strategy that partitions the input instance and applies tailored sampling algorithms to each partition, ensuring efficient and accurate estimation. This approach runs in $\O\left(\frac{N}{\sqrt{\OUT}}\right)$ time. In contrast, the previous result by the standard reduction to uniform sampling~\cite{chen2020random} requires $\O\left(\frac{N^2}{\OUT}\right)$ time. Consequently, our algorithm improves the existing result by a factor of $\O\left(\frac{N}{\sqrt{\OUT}}\right)$. 

    \noindent
    {\bf Section~\ref{sec:lb-matrix} (lower bounds for matrix query).} We establish a matching lower bound of $\Tilde{\Omega}\left(\frac{N}{\sqrt{\OUT}}\right)$ for both uniform sampling and approximate counting. First, we derive this bound via a reduction from the set disjointness problem in communication complexity to approximate counting. Notably, this lower bound is information-theoretic and thus applies to \emph{all} algorithms, including non-combinatorial approaches that might leverage fast matrix multiplication. 
    Following the reduction from approximate counting to uniform sampling \cite{chen2020random}, it translates to the same lower bound for uniform sampling, assuming $\O\left(\frac{N}{\sqrt{\OUT}}\right)$ preprocessing time, for all algorithms. Second, for combinatorial algorithms, we provide an alternative result by resorting to the reduction from query evaluation to uniform sampling \cite{deng2023join}. The intuition is that repeatedly drawing samples for $\O(\OUT)$ iterations would retrieve the full query result with high probability. Since any combinatorial algorithm for evaluating matrix query requires $\Omega\left(N \sqrt{\OUT}\right)$ time~\cite{amossen2009faster}, a sampling algorithm (assuming $O(N)$ preprocessing time) that improves to $O\left(\frac{N}{\OUT^{\frac{1}{2}+\gamma}}\right)$ for any small constant $\gamma>0$, would imply a combinatorial evaluation algorithm running in $\O\left(N \cdot \OUT^{\frac{1}{2} -\gamma}\right)$ time, which is impossible.

    \noindent {\bf Section \ref{appendix:star} (star query).} The upper and lower bounds derived for matrix query extend naturally to star queries. We achieve a tight time complexity of $\Theta\left(\frac{N}{\OUT^{1/k}}\right)$ for the $k$-star query. 

    \noindent
    {\bf Section~\ref{appendix:chain} (chain query).} We finally turn to chain queries. 
    For both the uniform sampling and approximate counting problems, we provide a new upper bound $O(N)$ by resorting to the \emph{k-minimum value summary}~\cite{cohen2008tighter, hu2020parallel,cohen2007summarizing}. Together with the existing upper bound $O(\frac{N^{\lceil \frac{k+1}{2}\rceil}}{\OUT})$ \cite{chen2020random}, we obtain a hybrid upper bound $O(\min\{N,\frac{N^2}{\OUT}\})$ for the $3$-chain query and $O(N)$ for general chain query with $k\ge4$. For the approximate counting problem, we prove a matching lower bound $\Omega(\min\{N,
    \frac{N^2}{\OUT}\})$ for the $3$-chain query and $\Omega(N)$ for general chain queries with $k \ge 4$. Following the reduction from approximate counting to uniform sampling \cite{chen2020random}, it translates to a lower bound of $\Omega(N)$ for general chain queries with $k \ge 4$ and for the $3$-chain query when $\OUT \le N$, assuming $O(N)$ preprocessing time. All these lower bounds rely on the information-theoretic hardness, thus apply to \emph{all} algorithms, including those that might leverage fast matrix multiplication.

    \section{Uniform Sampling over Matrix Query}\label{sec:sample-matrix}

    In this section, we present our uniform sampling algorithm for the basic matrix query. In Section~\ref{sec:sampling-matrix-framework}, we show a two-step framework of the sampling algorithm and also compare it with the prior framework. In Section~\ref{sec:step1} and Section~\ref{sec:step2}, we show detailed procedures to implement these two steps separately. At last, we give an analysis in Section~\ref{sec:sampling-matrix-analysis}.
        
    \begin{algorithm}[t]
        \caption{\textsc{SampleMatrix}$(\R)$}
        \label{alg:sample-matrix}
        \KwIn{An instance $\R$ for the matrix query $\matrix=\pi_{A,C} R_1(A,B) \Join R_2(B,C)$, with some auxiliary indices built for $\R$ if needed.}
        \KwOut{A uniform sample of the query result $\matrix(\R)$.}
        \While{\true}{
            $s\gets $\textsc{JoinSampleMatrix}$(\R)$; \quad \qquad \quad \quad \qquad \ \ \ \ \ \ \ \ \ \ \ \ \ $\blacktriangleright$ Algorithm~\ref{alg:heavy-sample} or Algorithm~\ref{alg:light-sample}\; 
            \lIf{$s \neq \perp$}{\textbf{break}}
        }
        \lIf{\textup{\textsc{MatrixAccept}}$(\R,s) = \true$}{\Return $(\pi_A s, \pi_C s)$ \ \ \   $\blacktriangleright$ \textup{Algorithm~\ref{alg:accept}}}
        \lElse{\Return $\perp$}
    \end{algorithm}

    \subsection{Framework}
    \label{sec:sampling-matrix-framework}
    As described in Algorithm~\ref{alg:sample-matrix}, 
    our framework is conceptually as simple as two steps. In {\bf step 1} (Lines 1-3), we sample a full join result $s = (a,b,c) \in R_1 \Join R_2$ uniformly at random. In {\bf step 2} (Lines 4-5), we accept $(a,c)$ as the final sample with probability $\frac{1}{\deg(a,c)}$, where 
        $\deg(a,c) = |\pi_B(R_1\ltimes a) \cap \pi_B(R_2 \ltimes c)|$,
    i.e., the number of full join results participated by $(a,c)$. 
    If no sample is returned (Line 1), we repeat these two steps until a successful sample is returned. We will formally prove why this framework can produce a uniform sample for $\matrix$ in Section~\ref{sec:sampling-matrix-analysis}. The intuition is that for each query result $(a,c)$, it will appear in the uniform sample of full join result with probability $\frac{\deg(a,c)}{\OUT_\Join}$, since $(a,c)$ will participate in $\deg(a,c)$ full join results, each with a distinct value in $(R_1 \ltimes a) \cap (R_2 \ltimes c)$. As $\deg(a,c)$ can be different for different query pairs $(a,c)$, the first step does not guarantee a uniform sample of the query result. Therefore, an additional acceptance step is required to rescale the probability. After the second step, every query result $(a,c)$ is accepted with probability $\frac{\deg(a,c)}{\OUT_\Join} \cdot \frac{1}{\deg(a,c)}= \frac{1}{\OUT_\Join}$, which only depends on the number of full join results.


     \begin{algorithm}[t]
        \caption{\textsc{JoinSampleMatrix-H}($\R$)}\label{alg:heavy-sample}
        \KwIn{An instance $\R$ of the matrix query $\matrix=\pi_{A,C} R_1(A,B) \Join R_2(B,C)$}
        \KwOut{A full join result of $R_1\Join R_2$.}
        \tcp{Suppose the following statistics are computed in the preprocessing step: For each value $b \in \adom(B)$,
            $W_b = |R_1 \ltimes b|\cdot |R_2 \ltimes b|$; and $W = \sum_{b \in \adom(B)} W_b$.}
        $b \gets$ a random sample from $\adom(B)$ with probability $\frac{W_b}{W}$\;
        $a \gets$ the $\randint(1,|R_1\ltimes b|)$-th value in $\pi_A (R_1 \ltimes b)$\;
        $c \gets$ the $\randint(1,|R_2\ltimes b|)$-th value in $\pi_C (R_2 \ltimes b)$\;
        \Return $(a,b,c)$\;
        \end{algorithm}

        \subsection{Step 1: Sample Full Join Result}
        \label{sec:step1}

        Throughout the paper, we will present two 
        different strategies of sampling a full join result from $R_1\Join R_2$. For uniform sampling over matrix query, Algorithm~\ref{alg:heavy-sample} would suffice; hence, we simply focus on this strategy now. Later, when using uniform sampling for approximate counting matrix query, we will introduce another strategy in Algorithm~\ref{alg:light-sample} and leave all the details to Section~\ref{sec:count-matrix-variant}.
        %

        \paragraph{Auxiliary indices} In addition to the operations supported by the property testing model, we need to build additional indices to support the following operations in $O(1)$ time: 
        \begin{itemize}[leftmargin=*]
            \item For each value $b\in \adom(B)$, return the ``weight'' $W_b = |R_1 \ltimes b| \cdot |R_2 \ltimes b|$.
            \item Return the total weight $W = \sum_{b \in \adom(B)} W_b$. 
            \item Return a random weighted sample $b$ from $\adom(B)$ with probability $\frac{W_b}{W}$.
        \end{itemize}
    
        As described in Algorithm~\ref{alg:heavy-sample}, \textsf{JoinSampleMatrix-H} first samples a value $b \in \adom(B)$ with probability proportional to its weight $W_b$. It then samples a full join tuple containing $b$ uniformly at random. More specifically, it independently samples an $A$-value $a$ from the neighbors of $b$ in $R_1$ and a $C$-value $c$ from the neighbors of $b$ in $R_2$ uniformly at random. Observe that $(a,b,c)$ is sampled with probability $\frac{W_b}{W} \cdot \frac{1}{|R_1 \ltimes b|} \cdot \frac{1}{|R_2\ltimes b|} = \frac{1}{W}$. Consequently, $(a,b,c)$ is a uniform sample from $R_1\Join R_2$.

           \begin{algorithm}[t]
        \caption{\textsc{MatrixAccept}($\R, (a, b, c)$)}\label{alg:accept}
        \KwIn{An instance $\R$ of $\matrix=\pi_{A,C} R_1(A,B) \Join R_2(B,C)$, and a full join result $(a,b,c)\in R_1(A,B) \Join R_2(B,C)$.}
        \KwOut{A Boolean value indicating acceptance.}
        $F \leftarrow 0$\;
        \lIf{$|R_1 \ltimes a| < |R_2 \ltimes c|$}{
            $S \gets \pi_B (R_1 \ltimes a)$}
        \lElse{
            $S \gets \pi_B (R_2 \ltimes c)$}
        \tcp{We sample from $S$ using uniform sampling without replacement\footnotemark.}
        \While{at least one element in $S$ is not visited yet}{
            $b' \leftarrow$ a random sample from non-visited elements in $S$\;
            \lIf{$b' = b$}{\Continue}
            \lIf{$(a,b') \notin R_1$\textup{\textbf{ or }}$(b',c) \notin R_2$}{$F \leftarrow F+1$ } 
            \lElse{\Break}
            }
            \lIf{$\randint(1, |S|) \le F+1$}{\Return \true}
        \Return \false\;
        \end{algorithm}
        \footnotetext{We can implement sampling without replacement efficiently using an online version of Fisher-Yates shuffle~\cite{fisher1953statistical,knuth1998art}, which has a cost linear only to the number of samples drawn instead of the total number of elements in the universe, which is $|S|$ in the context of Algorithm~\ref{alg:accept}. The procedure uses an initially empty hash map $H$ to store index swaps. Conceptually, $H[i]=i$ if no value pair has been stored in $H$ with key $i$. 
        To draw the $i$-th sample, we generate a uniform random integer $j$ from $\{i,i+1,\cdots, |S|\}$, and return the $H[j]$-th element in $S$. Meanwhile, we need to update $H$ by storing (key, value) pair $(j,i)$ into $H$ to effectively swap the chosen index with the one at the $i$-th position, ensuring it won't be sampled again.}

        \subsection{Step 2: Acceptance Check}
        \label{sec:step2}
        The main challenge is to implement the acceptance step. Recall that this step takes the full join tuple $(a,b,c)$ generated in Step 1 and decides whether to accept $(a,c)$ as a final sample. To ensure uniformity, we must accept $(a,c)$ with probability $\frac{1}{\deg(a,c)}$, where $\deg(a,c) = |\pi_B(R_1\ltimes a) \cap \pi_B(R_2 \ltimes c)|$. Explicitly computing $\deg(a,c)$ is too expensive ($O(N)$ worst-case). Instead, we simulate this probability using the following helper toolkit:
        \begin{proposition} \label{prop:simulation}
        Let $Y$ be a discrete random variable taking values in $[M]$, where $M$ is a positive integer. Let $U$ be a uniform random variable on $[M]$, independent of $Y$. Then $\Pr[U \le Y] = \frac{1}{M} \cdot \expt[Y]$.
        \end{proposition}
         \begin{proof}[Proof of Proposition~\ref{prop:simulation}]
            By the Law of Total Probability and the independence of $U$ and $Y$:
            $\Pr[U \le Y] = \sum_{y=1}^M \Pr[U \le y] \cdot \Pr[Y=y].$
            Since $U$ is uniform on $[M]$, $\Pr[U \le y] = \frac{y}{M}$. Substituting this yields:
            $\Pr[U \le Y] = \frac{1}{M} \sum_{y=1}^M y \cdot \Pr[Y=y] = \frac{\expt[Y]}{M}$.
        \end{proof}

        To apply this proposition, we first identify the smaller set of candidate neighbors, $S$. Wlog, assume $|\pi_B(R_1 \ltimes a)| \le |\pi_B(R_2 \ltimes c)|$, so we set $S = \pi_B(R_1 \ltimes a)$ and $M = |S|$. Note that $b \in S$ is already known to be a valid join value (a ``success'') because $(a,b,c)$ was sampled from the join.
        
        We construct the random variable $Y$ by sampling from the \emph{other} elements in $S$. Specifically, we sample without replacement from $S - \{b\}$ and count the number of ``failures'' ($F$) encountered before finding another valid join value $b' \in \pi_B(R_1 \ltimes a) \cap \pi_B(R_2 \ltimes c)$. If no other valid join value exists (i.e., $b$ is unique), we continue until $S - \{b\}$ is exhausted. We define $Y = F+1$.
        
        Algorithm~\ref{alg:accept} implements this procedure. It iterates through random samples $b'$ from $S$. If $b' = b$, it is skipped (ensuring we sample from $S - \{b\}$). If $b' \neq b$, we check if $b'$ connects $a$ and $c$ (i.e., $(b', c) \in R_2$). If it does not, we increment the failure count $F$. The process stops as soon as a valid $b'$ is found or $S$ is exhausted.
        
        As shown in Section~\ref{sec:sampling-matrix-analysis}, the expectation of this specific construction is $\expt[Y] = \frac{|S|}{\deg(a,c)}$. Finally, we generate a random integer $X$ uniformly from $[|S|]$. If $X \le F+1$, we return \true{}; otherwise, we return \false{}. By Proposition~\ref{prop:simulation}, the acceptance probability is exactly $\frac{1}{|S|} \cdot \frac{|S|}{\deg(a,c)} = \frac{1}{\deg(a,c)}$.
        \subsection{Analysis}
        \label{sec:sampling-matrix-analysis}
        We now prove the correctness and analyze the time complexity of Algorithm~\ref{alg:sample-matrix}. First, we focus on Algorithm~\ref{alg:accept} and introduce the concept of the {\em negative hypergeometric distribution}:
        
        \begin{definition}[Negative Hypergeometric Distribution~\cite{johnson2005univariate}]
            \label{def:neg-hypergeo}
            Consider an experiment of drawing balls without replacement from a box containing $P$ balls, of which $K$ are successes. Let a random variable $Z \sim \texttt{NHG}(r,P,K)$ count the number of failures drawn before obtaining the $r$-th success. Its expectation is $\expt[Z] = \frac{r(P-K)}{K+1}$.
        \end{definition}

        In Algorithm~\ref{alg:accept}, when determining whether to accept the query result $(a,c)$ that is generated from the full join result $(a,b,c)$, we are sampling from a population of $P=|S|-1$ distinct values (excluding $b$) with only $K=\deg(a,c)-1$ successes to retrieve the first ($r=1$) one. Thus, $F \sim \texttt{NHG}(1, |S|-1, \deg(a,c)-1)$. The expectation of $F$ is:
        $\expt[F] = \frac{|S|}{\deg(a,c)} - 1.$ Now we apply Proposition~\ref{prop:simulation}. Let $Y = F+1$. We know that $Y\in[|S|]$ and have established that $\expt[Y] = \frac{|S|}{\deg(a,c)}$. Algorithm~\ref{alg:accept} returns \true{} if a random integer 
        $U \in [|S|]$ satisfies $U \le Y$.
        By Proposition~\ref{prop:simulation}, this occurs with probability:
        $\Pr[\text{accept}] = \frac{\expt[Y]}{|S|} = \frac{1}{|S|} \cdot \frac{|S|}{\deg(a,c)} = \frac{1}{\deg(a,c)}.$
        
        We now return to Algorithm~\ref{alg:sample-matrix}. Consider an arbitrary query result $(a,c) \in \matrix(\R)$. In \textbf{one iteration} of the loop, the probability that the sampled full join result projects to $(a,c)$ is $\frac{\deg(a,c)}{\OUT_\Join}$. In \textbf{Step 2}, such a candidate is accepted with probability $\frac{1}{\deg(a,c)}$. Consequently, in a single iteration, $(a,c)$ is returned with probability $\frac{\deg(a,c)}{\OUT_\Join} \cdot \frac{1}{\deg(a,c)} = \frac{1}{\OUT_\Join}$. Since this probability is the same for any $(a,c) \in \matrix(\R)$, the final output (conditioned on acceptance) is a uniform sample of $\matrix(\R)$.
    
        We now analyze the time complexity of Algorithm~\ref{alg:sample-matrix}. The preprocessing step takes $O(N)$ time. After preprocessing, each invocation of Algorithm~\ref{alg:heavy-sample} takes $O(1)$ time. As Algorithm~\ref{alg:heavy-sample} always returns a successful sample (without rejection), the while-loop at 
        Lines 1-3 
        takes $O(1)$ time. The invocation of Algorithm~\ref{alg:accept} at
        Line 4
        takes $O(F+1)$ time. From our analysis above, the expectation $\expt[F+1] = \frac{|S|}{\deg(a,c)} = \frac{\min\{|R_1\ltimes a|,|R_2\ltimes c|\}}{\deg(a,c)}$. Hence, Algorithm~\ref{alg:sample-matrix} takes:
        \begin{align*}
        \expt[\text{time per invocation}] &= 1+ \sum_{(a,c)\in\matrix(\R)} \frac{\deg(a,c)}{\OUT_\Join} \cdot \frac{\min\{|R_1\ltimes a|,|R_2\ltimes c|\}}{\deg(a,c)}
        \le 1 + \frac{N\sqrt{\OUT}}{\OUT_\Join}.
        \end{align*}
        expected time (omitting the big-$O$), where the last inequality follows from the fact that $\min\{x,y\} \le \sqrt{xy}$ and the Cauchy-Schwarz inequality\footnote{
    $\sum_{(a,c)\in\matrix(\R)} \min\{|R_1\ltimes a|,|R_2\ltimes c|\} \le \sum_{(a,c)\in\matrix(\R)} \sqrt{|R_1\ltimes a|\cdot|R_2\ltimes c|} \le N \sqrt{\OUT}.$}. 
        Algorithm~\ref{alg:sample-matrix} succeeds with probability $\frac{\OUT}{\OUT_\Join}$, so the expected number of invocations before getting one successful sample is $\frac{\OUT_\Join}{\OUT}$. The total expected cost is the product of these two terms:
        $O\left(\frac{\OUT_\Join}{\OUT} \cdot \left(1 + \frac{N\sqrt{\OUT}}{\OUT_\Join}\right)\right) = O\left(\frac{\OUT_\Join}{\OUT} + \frac{N}{\sqrt{\OUT}}\right) = O\left(\frac{N}{\sqrt{\OUT}}\right)$,
        where the last inequality follows from the fact in~\cite{amossen2009faster} that for $\matrix$, $\OUT_\Join \le N\sqrt{\OUT}$.
        
        \begin{theorem}\label{thm:sample-matrix-ub}
            For $\matrix$, there is an algorithm that can take as input an arbitrary instance $\R$ of input size $N$ and output size $\OUT$, and builds an index in $O(N)$ time such that a uniform sample from the query result $\matrix(\R)$ can be returned in $O\left(\frac{N}{\sqrt{\OUT}}\right)$ expected time. 
        \end{theorem}

        \paragraph{Standard Amplification Techniques} Specifically, if we set a time budget of twice the expected runtime and restart the algorithm if it does not terminate within that budget, the probability of failure drops exponentially with the number of restarts. Repeating this process $O(\log N)$ times guarantees the result with high probability. This observation extends to our other sampling algorithms.
       
                
        \subsection{A variant of Step 1: Sample Full Join Result}
        \label{sec:count-matrix-variant}
        \begin{algorithm}[t]
        \caption{\textsc{JoinSampleMatrix-L}($\R$)}\label{alg:light-sample}
        \KwIn{An instance $\R$ of the matrix query $\matrix=\pi_{A,C} R_1(A,B) \Join R_2(B,C).$}
        \KwOut{A full join result of $R_1\Join R_2$ or $\perp$ indicating ``failure''.}
        \tcp{Suppose the following statistic is computed in the preprocessing step: $\Delta \ge \max_{b \in \adom(B)} |R_2\ltimes b|$.}
        $(a, b) \leftarrow $ a uniform sample from $R_1$\;
        \lIf{$|R_2\ltimes b| = 0$}{\Return $\perp$}
        $j \leftarrow \randint(1,|R_2 \ltimes b|)$\;
        $c \leftarrow$ the $j$-th value in $\pi_{C} (R_2\ltimes b)$\;
        \lIf{$\randint(1, \Delta) > |R_2\ltimes b|$}{\label{step:smooth}\Return $\perp$}
        \Return $(a,b,c)$\;
        \end{algorithm}
        To prepare for the next section on approximately counting matrix queries, we present an alternative strategy for {\bf Step 1} to sample a full join result uniformly at random. This strategy does not require preprocessing to build auxiliary indices; instead, it relies on an upper bound $\Delta$ on the maximum degree of $B$-values in one of the relations. Wlog, we assume $\Delta \ge \max_{b \in \adom(B)} |R_2 \ltimes b|$.

        As outlined in Algorithm~\ref{alg:light-sample}, \textsf{JoinSampleMatrix-L} employs rejection sampling. It first samples a tuple $(a,b)$ uniformly at random from $R_1$, and subsequently samples a value $c$ uniformly at random from the neighbors of $b$ in $R_2$ (i.e., from $\pi_C(R_2 \ltimes b)$). A specific candidate join result $(a,b,c)$ is initially proposed with probability $\frac{1}{|R_1|} \cdot \frac{1}{|R_2\ltimes b|}$. However, this distribution is not uniform, as the degree $|R_2\ltimes b|$ varies across different $B$-values. To correct for this bias, the algorithm accepts the candidate $(a,b,c)$ with probability $\frac{|R_2\ltimes b|}{\Delta}$. This is a valid probability since $\Delta \ge |R_2 \ltimes b|$. Consequently, the probability of returning a specific tuple $(a,b,c)$ becomes $\frac{1}{|R_1|} \cdot \frac{1}{|R_2\ltimes b|} \cdot \frac{|R_2\ltimes b|}{\Delta} = \frac{1}{|R_1| \cdot \Delta}$, which is uniform across all full join results. 

        The algorithm performs a constant number of primitive operations, so it runs in $O(1)$ time.
        For any specific full join tuple $t = (a,b,c) \in R_1 \bowtie R_2$, the probability that it is returned is the product of sampling $(a,b)$ from $R_1$, sampling $c$ from $R_2 \ltimes b$, and passing the rejection check:
        $
            \Pr[\text{return } t] = \frac{1}{|R_1|} \cdot \frac{1}{|R_2 \ltimes b|} \cdot \frac{|R_2 \ltimes b|}{\Delta} = \frac{1}{|R_1| \cdot \Delta}.
        $
        Since this probability is independent of $t$, the output distribution is uniform conditioned on not returning $\perp$. Summing this probability over all $t \in R_1 \bowtie R_2$ gives the total success probability $\Pr[\text{success}] = \frac{\OUT_\Join}{|R_1| \cdot \Delta}$, which implies $\Pr[\text{return } \perp] = 1 - \frac{\OUT_\Join}{|R_1| \cdot \Delta}$.
        
        \begin{theorem}
        \label{thm:light-sample-matrix}
            Algorithm~\ref{alg:light-sample} runs in $O(1)$ time. It returns $\perp$ with probability $1 - \frac{\OUT_\Join}{|R_1|\cdot \Delta}$. Conditioned on not returning $\perp$, the output $(a,b,c)$ is a uniform sample from the full join result $R_1 \Join R_2$.
        \end{theorem}

        
        To distinguish between the two sampling strategies, we denote the version utilizing Algorithm~\ref{alg:heavy-sample} as {\sc SampleMatrix-H} and the version utilizing Algorithm~\ref{alg:light-sample} as {\sc SampleMatrix-L}.
        
        \section{Approximate Counting for Matrix Query}
        \label{sec:count-matrix}
        In this section, we present our approximate counting algorithm for the basic matrix query. We show the framework of a hybrid strategy of uniform sampling algorithms in Section~\ref{sec:count-matrix-framework} and the detailed algorithms in Section~\ref{sec:count-matrix-algorithm}. All missing materials are provided in Section~\ref{sec:count-matrix-analysis}.

        \subsection{Framework}
        \label{sec:count-matrix-framework}
        \citet{chen2020random} showed a reduction from approximate counting to uniform sampling for join-project queries. Following Proposition \ref{prop:reduction}, Theorem~\ref{thm:sample-matrix-ub} immediately yields an algorithm that, with constant probability, returns an \(\epsilon\)-approximation of \(\OUT\) in expected time $O\left(N + \frac{N}{\epsilon^2 \cdot \sqrt{\OUT}}\right) = O(N)$, by the following facts: (1) Algorithm~\ref{alg:sample-matrix} returns a query result with probability $\frac{\OUT}{\OUT_\Join}$; (2)
        Algorithm~\ref{alg:sample-matrix} runs in $O(1+\frac{N\sqrt{\OUT}}{\OUT_\Join})$ expected time. 
    However, in the property-testing model, we seek a truly sublinear algorithm whose cost decreases as the output size $\OUT$ grows. This motivates a refined reduction from approximate counting to uniform sampling that explicitly leverages $\OUT$.

        \begin{proposition}[\cite{chen2020random}]
        \label{prop:reduction}
         Suppose there is a uniform sampling algorithm $\mathcal{A}$ that can successfully return a query result with probability $\frac{\OUT}{W}$. If repeating $Y$ independent invocations of $\mathcal{A}$ and ending up with $X$ successes, then $\frac{X\cdot W}{Y}$ is an unbiased estimator of $\OUT$. If $Y = \Omega(\frac{W}{\epsilon^2 \cdot \OUT})$, $\frac{X\cdot W}{Y}$ is an $\epsilon$-approximation of $\OUT$ with constant probability. 
         \end{proposition}

    Our new reduction employs a hybrid strategy using two distinct instantiations of Algorithm~\ref{alg:sample-matrix}. Specifically, we partition the values in attribute $B$ into \emph{heavy} and \emph{light} sets based on their degrees in the input relations. This effectively divides the full join space into two parts. For query results witnessed by heavy $B$-values, we invoke {\sc SampleMatrix-H}. Recall that {\sc SampleMatrix-H} requires precomputed statistics for all $B$-values; this is feasible here because the number of heavy values is bounded due to the input size constraint. For query results witnessed by light $B$-values, we invoke {\sc SampleMatrix-L}. Recall that {\sc SampleMatrix-L} requires a prior upper bound on the degrees, which is naturally provided by the threshold defining the light set. Since the sets of projected query results from these two parts may overlap, our final estimator combines them using a specialized intersection estimation technique. Implementing this idea entails several key challenges:

    \begin{itemize}[leftmargin=*]
        \item First, we do not know $\OUT$ unless we get an estimator for it, but our hybrid strategy depends on it. To overcome this barrier, we adopt the common strategy of geometric search for $\OUT$ from $N^2, \frac{N^2}{2}, \frac{N^2}{4}$ and so on, until we can get an estimator (almost) matching our guess. 
        \item Second, precisely partitioning all $B$-values into heavy and light sets is infeasible, as computing the degree of every value requires $O(N)$ time. To overcome this bottleneck, we employ random sampling. Specifically, by drawing a subset of $\O(N/\Delta)$ random tuples (with replacement) from $R_2$, we ensure that every $B$-value with degree at least $\Delta$ is detected with high probability. We then verify the true degree of each detected candidate to confirm if it is a heavy $B$-value.
        \item Third, by Proposition \ref{prop:reduction}, we can obtain accurate estimators for the query results in these two parts separately. However, the cost of each estimator scales inversely with the true result size of its respective part. If one part is too small, a blind application would blow up the cost. We therefore seek an approximate counting protocol with a limited set of calibrated guesses: in the ideal case (neither part too small), it yields accurate estimates for both; otherwise, it returns a safe upper bound without significant loss in overall accuracy. 
        \item Finally, we combine the two estimators. Since the heavy and light parts may produce overlapping query results, we must estimate the size of their intersection to avoid double counting. As neither result set is materialized, naive sampling and membership testing are too expensive. Therefore, we require a more efficient intersection estimation technique. Furthermore, we must account for the error inherent in this intersection estimate when constructing the final estimator.
    \end{itemize}

    \subsection{Algorithm}
    \label{sec:count-matrix-algorithm}
    \begin{algorithm}[t]\caption{$\textsc{ApproxCountMatrix}_{\epsilon,\delta}(\R)$}
        \label{alg:approx-counting-matrix}
         \KwIn{An instance $\R$ of $\matrix$, approximation quality $\epsilon$ and error $\delta$.}
    \KwOut{An estimate of $|\matrix(\R)|$.}
    $\Lambda \gets N^2$, $k \gets \log \frac{2}{\delta}$\;
    \While{$\Lambda \ge 1$}{
    \lFor{$i \in [k]$}{$s_i \gets \textsc{ApproxCountMatrixWithGuess}_{\epsilon, \delta/2}(\R, \Lambda)$; \qquad\ \ \   $\blacktriangleright$ \textup{Algorithm~\ref{alg:approx-counting-matrix-with-guess}}}
    $s \gets $ the median of $s_1,s_2,\cdots,s_k$\;
    \lIf{$s \ge \Lambda$}{\Return $s$}
    $\Lambda \gets \frac{\Lambda}{2}$\;
    }
    \Return $0$\;
    \end{algorithm}

    \begin{algorithm}[t]\caption{$\textsc{ApproxCountMatrixWithGuess}_{\epsilon,\delta}(\R, \Lambda)$}\label{alg:approx-counting-matrix-with-guess}
    \KwIn{An instance $\R$ of the matrix query $\matrix=\pi_{A,C} R_1(A,B) \Join R_2(B,C)$, parameter $\Lambda \in [1,N^2]$, approximation quality $\epsilon$ and error $\delta$.}
    \KwOut{An estimate of $|\matrix(\R)|$.}
    $\epsilon' \gets \epsilon/5$\;
    $B^\heavy \gets \textsc{DetectHeavy}_{ \delta/4}(\R,\Lambda)$; \qquad \qquad \qquad \qquad  \qquad \qquad  \qquad \qquad \ \ \ \ $\blacktriangleright$ \textup{Algorithm~\ref{alg:heavy-detect}}\;
    $\R^\heavy \gets \{\sigma_{B\in B^\heavy}R_1, \sigma_{B \in B^\heavy}R_2\}$; \ \ \ \ \ \ \ \ \ \qquad   \tcp{conceptual; simulated by filtering}
    $\R^\light \gets \{R_1 -\sigma_{B\in B^\heavy}R_1, R_2- \sigma_{B \in B^\heavy}R_2\}$;   \ \tcp{conceptual; simulated by filtering}
     $\widehat{\OUT}_\Join^\heavy \leftarrow \sum_{b \in B^\heavy}|R_1 \ltimes b| \cdot |R_2\ltimes b|$, $\widehat{\OUT}_\Join^\light \leftarrow N \cdot \sqrt{\Lambda}$\;
    $s^\heavy \gets \textsc{ApproxCountWithThreshold}_{\epsilon', \delta/4}\left(\R^\heavy, \widehat{\OUT}_\Join^\heavy,\frac{\epsilon'\Lambda}{16}\right)$;\ \ \ \ \ \ \ \ $\blacktriangleright$ \textup{Algorithm~\ref{alg:approximate-count-with-threshold}}\;
    $s^\light \gets \textsc{ApproxCountWithThreshold}_{\epsilon', \delta/4}\left(\R^\light, \widehat{\OUT}_\Join^\light,\frac{\epsilon'\Lambda}{16}\right)$; \ \ \ \ \ \ \ \ \ \ \ \ $\blacktriangleright$ \textup{Algorithm~\ref{alg:approximate-count-with-threshold}}\;
    \If{$s^\heavy \neq \perp$ and $s^\light \neq \perp$}{
        $\beta \gets \textsc{IntersectEstimate}_{\epsilon', \delta/4}(\R,\R^\heavy, \R^\light)$;\qquad \qquad\ \ \ \ \ \ \ \ \ \ \ \  \quad \qquad  $\blacktriangleright$ \textup{Algorithm~\ref{alg:intersect-estimate}}\;
        \Return $s^\heavy + s^\light - \min\{s^\heavy, s^\light \cdot \beta\}$\;
    }
    \lElseIf{$s^\heavy \neq \perp$ and $s^\light = \perp$}{\Return $s^\heavy$}
    \lElseIf{$s^\heavy =\perp$ and $s^\light \neq \perp$}{\Return $s^\light$}
    \lElse{\Return $0$}
    \end{algorithm}

    {\bf \textsc{ApproxCountMatrix}.} Our outermost procedure is presented in Algorithm \ref{alg:approx-counting-matrix}. To address the first challenge, we try a sequence of guesses for the true output size, \(\Lambda = N^2, \tfrac{N^2}{2}, \tfrac{N^2}{4}, \ldots, 1\). For each guess $\Lambda$, we invoke our core routine, \textsc{ApproxCountMatrixWithGuess}, to obtain an estimator conditioned on $\Lambda$. To amplify the success probability (as shown in our analysis), we repeat this estimation process $O(\log \frac{1}{\delta})$ times and take the median of these estimators. Given an estimator $s$ based on the guess $\Lambda$, we compare it with $\Lambda$. If $s$ exceeds $\Lambda$, we return it immediately; otherwise, we halve $\Lambda$ and proceed to the next guess. This process continues until a valid estimate is found or all guesses have been exhausted.

    \paragraph{\textsc{ApproxCountMatrixWithGuess}} As described in Algorithm~\ref{alg:approx-counting-matrix-with-guess}, it starts with identifying the heavy values in $B$, using a primitive called {\sc DetectHeavy}. A value $b \in \adom(B)$ is {\em heavy} if it appears in more than $\sqrt{\Lambda}$ tuples in $R_2$, and {\em light} otherwise. Since with high probability {\sc DetectHeavy} correctly finds all heavy values, we simply assume it does so. We denote the set of heavy values as $B^\heavy$.
    We {\em conceptually} use $B^\light = \adom(B) - B^\heavy$ to denote the light values, but we do not explicitly list $B^\light$, because there might be too many of them. For the same reason, we {\em conceptually} partition $\R$ into two sub-instances (whose query results may overlap): $\R^\heavy$ contains all tuples with heavy $B$-values and $\R^\light$ contains the remaining input tuples. For $\R^\heavy$, we can exactly calculate the number of full join results. For $\R^\light$, we know that the number of full join results is at most $N\sqrt{\Lambda}$, since there are at most $N$ tuples in $R_1$, and each of them joins with at most $\Lambda$ tuples in $R_2$. That is essentially how we set up  $\widehat{\OUT}^\heavy_\Join$ 
    and $\widehat{\OUT}^\light_\Join$ as the upper bounds of full join results in each part. 
  
    Then, we rely on a separate procedure, {\sc ApproxCountWithThreshold}, to estimate the number of query results in each part. This procedure works like this: if the actual count is larger than a predefined threshold $\tau$, it successfully returns an $\epsilon$-approximation estimator $s$. If the procedure reports a failure ($\perp$), it means the actual count must be less than the threshold $\tau$. Using the feedback ($s^\heavy, s^\light$) on the two parts from {\sc ApproxCountWithThreshold}, we combine them into a final estimator. Let $\OUT^?$ be the number of query results produced by $\R^?$, where $?$ is either heavy or light. We further distinguish the three cases. 

    \begin{itemize}[leftmargin=*]
    \item {\bf Case 1: both estimates are successful ($s^\heavy\neq \perp$ and $s^\light\neq \perp$).} By the design of {\sc ApproxCountWithThreshold}, $s^\heavy$ and $s^\light$ are accurate estimates of $\OUT^\heavy$ and $\OUT^\light$, respectively, and both counts exceed the threshold. Since the result sets $\matrix(\R^\heavy)$ and $\matrix(\R^\light)$ may overlap, their sum would double-count the intersection. We employ the procedure {\sc IntersectEstimate} to estimate $\beta$, defined as the fraction of $\matrix(\R^\light)$ that is also present in $\matrix(\R^\heavy)$. Consequently, $s^\light \cdot \beta$ estimates the size of this intersection. However, stochastic errors in estimating $\beta$ could yield an intersection estimate larger than $\OUT^\heavy$, which would destabilize the final result. To mitigate this, we leverage the fact that the intersection size cannot exceed $\OUT^\heavy$, yielding the robust combined estimator $s^\heavy + s^\light - \min\{s^\light \cdot \beta, s^\heavy\}$.
    
    \item {\bf Case 2: Only one estimate fails.} 
    Without loss of generality, assume $s^\light =\perp$. This failure implies that $\OUT^\light$ is negligible (below the threshold). Consequently, $s^\heavy$, which accurately estimates $\OUT^\heavy$, serves as a sufficient approximation for the total output size $\OUT$.
    
    \item {\bf Case 3: Both estimates fail ($s^\heavy = \perp$ and $s^\light = \perp$).} 
    These failures indicate that both $\OUT^\heavy$ and $\OUT^\light$ are below the threshold, implying that their union is at most $\frac{\epsilon'\Lambda}{8}$. As this is significantly smaller than the current guess $\Lambda$, we conclude that $\Lambda$ is an overestimate and return $0$ (signaling the algorithm to proceed to the next guess).
    \end{itemize}

       \begin{algorithm}[t]
        \caption{$\textsc{DetectHeavy}_{\delta}(\R, \Lambda)$}
        \label{alg:heavy-detect}
        \KwIn{An instance $\R$ of $\matrix=\pi_{A,C} R_1(A,B) \Join R_2(B,C)$, a parameter $\Lambda$ and error $\delta$.}
        \KwOut{A subset of values in $\adom(B)$ that appears in more than $\sqrt{\Lambda}$ tuples in $R_2$.}
        $k \gets \frac{N\cdot \log \frac{N}{\delta}}{\sqrt{\Lambda}}$, $B^\heavy \leftarrow \emptyset$\;
        \For{$i \in [k]$}{
        $(b,c) \leftarrow$ a sample drawn from $R_2$ uniformly at random\;
        \lIf{$|R_2\ltimes b| > \sqrt{\Lambda}$}{
            $B^\heavy \gets B^\heavy \cup \{b\}$\label{heavy-check}}
    }
    \Return $B^\heavy$\;
    \end{algorithm}

    \noindent {\bf \textsc{DetectHeavy.}} As described in Algorithm~\ref{alg:heavy-detect}, it randomly samples $\O(\frac{N}{\sqrt{\Lambda}})$ tuples from $R_2$, and collects $B$-values of all sampled tuples as the set $B^\heavy$. Also, it checks the true degree of values in $B^\heavy$, so that $B^\heavy \subseteq \{b \in \adom(B): |R_2 \ltimes b| \ge \sqrt{\Lambda}\}$. Hence, we can bound the size of $B^\heavy$ as $O(\frac{N}{\sqrt{\Lambda}})$. As proved later, with high probability, every $B$-value with degree larger than $\sqrt{\Lambda}$ in $R_2$ is detected, hence $B^\heavy=\{b \in \adom(B): |R_2 \ltimes b| \ge \sqrt{\Lambda}\}$.

      \begin{algorithm}[t]    
    \caption{$\textsc{ApproxCountWithThreshold}_{\epsilon,\delta}(\R,\widehat{\OUT}_\Join,\tau)$}
    \label{alg:approximate-count-with-threshold}
    \KwIn{An instance $\R$ of $\matrix = \pi_{A,C} R_1(A,B) \Join R_2(B,C)$, an upper bound $\widehat{\OUT}_\Join$ of the full join size $|R_1\Join R_2|$ 
    and threshold $\tau$.} 
    \KwOut{An estimator of $|\matrix(\R)|$ or $\perp$ indicating failure.}
        
    $\lambda \gets N^2$\label{initial}, $k_\delta \gets 3\log \frac{2}{\delta}$\;
    \While{$\lambda \ge \tau$\label{threshold}}{
       
    \For{$i \in [k_\delta]$}{               
        $Z \leftarrow  0$, $k_\lambda \leftarrow \frac{6\widehat{\OUT}_\Join}{\epsilon^2\lambda}$\;
        \For{$j \in [k_\lambda]$}{
            \tcp{\textsc{SampleMatrix-H} for $\R^\heavy$ and \textsc{SampleMatrix-L} for $\R^\light$;}
            $s \gets\textsc{SampleMatrix}\left(\R\right)$;\qquad \qquad \qquad \qquad \qquad \qquad \qquad \qquad $\blacktriangleright$ \textup{Algorithm~\ref{alg:sample-matrix}}\;
            \lIf{$s \neq \perp$}{
                $Z \leftarrow  Z+\widehat{\OUT}_\Join$
                }
        }
         $s_i \gets \frac{Z}{k_\lambda}$\;     
        }
        $s \gets $ the median of $s_1, s_2, \cdots, s_{k_\delta}$\;
        \lIf{$s \ge \lambda$\label{break}}{\Return $s$}
        $\lambda \gets \frac{\lambda}{2}$\label{halve}\;
        }
        \Return $\perp$\;
    \end{algorithm}
    
    \paragraph{Simulating Heavy and Light Sub-instances} 
    As $\R^\heavy$ and $\R^\light$ are not materialized, we simulate access to them by reusing the original indices and storing $B^\heavy$ in a hash table to support membership checks.
    When executing \textsc{SampleMatrix} on $\R^\light$, we enforce the following:
    \begin{itemize}[leftmargin=*]
        \item In \textbf{Step 1} (Algorithm~\ref{alg:light-sample}), if a tuple $(a,b)$ is sampled from $R_1$ such that $b \in B^\heavy$, we treat it as a rejection (return $\perp$). This results in a sampling probability of $\frac{1}{N\Delta}$ for each full join result, 
        but this still suffices for our analysis. 
        \item In \textbf{Step 2} (Algorithm~\ref{alg:accept}), when sampling a neighbor $b'$ from the neighbor list $S = \pi_B (\sigma_{B \notin B^\heavy}(R_1\ltimes a))$, 
        we simply sample values from $\pi_B (R_1\ltimes a)$ and skip any $b'$ if $b' \in B^\heavy$. As analyzed later, this won't increase our complexity asymptotically.
    \end{itemize}
    When executing \textsc{SampleMatrix} on $\R^\heavy$, we enforce the following:
     \begin{itemize}[leftmargin=*]
        \item In \textbf{Step 1} (Algorithm~\ref{alg:heavy-sample}) We simply focus on $B^\heavy$ instead of $\adom(B)$. More specifically, we compute $W_b$ for each $b \in B^\heavy$ and $W = \sum_{b \in B^\heavy} W_b$ in the preprocessing step. 
        \item In \textbf{Step 2} (Algorithm~\ref{alg:accept}), when sampling a neighbor $b'$ from the neighbor list $S = \pi_B (\sigma_{B \in B^\heavy}(R_1\ltimes a))$, 
        we simply sample values from $\pi_B (R_1\ltimes a)$ and skip any $b'$ if $b' \notin B^\heavy$. As analyzed later, this won't increase our complexity asymptotically.
    \end{itemize}
    All these adaptations make the probabilities correct while still having nice time complexities.
    
    \paragraph{\textsc{ApproxCountWithThreshold}} As described in Algorithm~\ref{alg:approximate-count-with-threshold}, the high-level idea is very similar to Algorithm~\ref{alg:approx-counting-matrix}. We try a sequence of guesses for the true output size, \(\lambda = N^2, \tfrac{N^2}{2}, \tfrac{N^2}{4}, \ldots, \tau\), where $\tau$ is the smallest guess we can try. This threshold is enforced to satisfy the time complexity constraints. Note that this guess $\lambda$ is independent of the guess $\Lambda$ in the outermost procedure Algorithm~\ref{alg:approx-counting-matrix}, since this is for estimating $\OUT^\heavy, \OUT^\light$ separately, which can be very different from $\OUT$. For each guess $\lambda$, we simply apply the reduction in Proposition~\ref{prop:reduction} as Lines 3-9. Recall that the {\sc MatrixSample} procedure returns a uniform sample of the query result from $\matrix(\R)$ with probability 
    $\frac{\OUT}{\widehat{\OUT}_\Join}$.
    Repeating this procedure $O(\frac{\widehat{\OUT}_\Join}{\epsilon^2 \lambda})$ independent times should give us a constant-approximation with at least constant probability when $\lambda = O(\OUT)$. To amplify the success probability (as shown in our analysis), we repeat this estimation process $O(\log \frac{1}{\delta})$ times and take the median of these estimators. If $s$ exceeds $\lambda$, we return it immediately; otherwise, we halve $\lambda$ and proceed to the next guess. This process continues until a valid estimate is found or all guesses have been exhausted.

     \begin{algorithm}[t]
    \caption{$\textsc{IntersectEstimate}_{\epsilon, \delta}(\R, \R^\heavy, \R^\light)$}
        \label{alg:intersect-estimate}
        \KwIn{An instance $\R$ of the matrix query $\matrix=\pi_{A,C} R_1(A,B) \Join R_2(B,C)$, and a partition of $\R = (\R^\heavy, \R^\light)$ based on a partition of attribute $B = (B^\heavy, B^\light)$.}
        \KwOut{An estimator of $\frac{|\matrix(\R^\heavy) \cap \matrix(\R^\light)|}{|\matrix(\R^\light)|}$.}
        $S,S'\gets \emptyset$\;
        \While{$|S| \le \frac{1}{2\epsilon^2} \log \frac{2}{\delta}$}{
            \While{\textup{\textbf{true}}}{$s \gets \textsc{SampleMatrix-L}(\R^\light)$; \qquad \qquad \qquad \qquad \qquad \qquad \ \ \ \qquad $\blacktriangleright$\
            \textup{Algorithm~\ref{alg:sample-matrix}}\;
        \lIf{$s \neq \perp$}{$S \gets S \cup \{s\}$ and \textbf{break}}
        }

        \For{each $b\in B^\heavy$}{
            \lIf{$(\pi_A s,b) \in R_1$ and $(b,\pi_C s) \in R_2$}{$S' \gets S'\cup \{s\}$}
            } 
        }
        \Return $\frac{|S'|}{|S|}$\;
    \end{algorithm}

    \paragraph{\textsc{IntersectEstimate}} As described in Algorithm~\ref{alg:intersect-estimate}, it takes a set ($S$) of uniform samples from $\matrix(\R^\light)$ through Algorithm~\ref{alg:sample-matrix}, and then checks the subset $(S')$ of sampled query results that are also part of the query results \(\matrix(\R^\heavy)\), by iterating through every value in $B^\heavy$. When taking a uniform sample from $\matrix(\R^\light)$, it hits the overlap with probability \(\frac{|\matrix(\R^\heavy) \cap \matrix(\R^\light)|}{\OUT^\light}\). Hence, \(\frac{|S'|}{|S|} \cdot \OUT^\light\) 
    serves as an unbiased estimator of the overlap. As we prove later, as long as a sufficiently large number of query results are sampled from \(\matrix(\R^\light)\), we can ensure an $\epsilon$ additive error on the approximation with probability at least $1-\delta$.

     \subsection{Analysis}
    \label{sec:count-matrix-analysis}
    
    We begin by analyzing the approximation quality of our counting algorithm by establishing the key properties of its helper procedures. Specifically, we prove Lemmas \ref{lem:heavy-detect} through \ref{lem:approx-counting-matrix-with-guess}, and finally derive Lemma \ref{lem:approx-counting-matrix} as a natural culmination.
    \begin{lemma}
        \label{lem:heavy-detect}
        In Algorithm~\ref{alg:heavy-detect}, with probability at least $1-\delta$, $B^\heavy = \{b \in \adom(B): |R_2\ltimes b| \ge \sqrt{\Lambda}\}$.
    \end{lemma}

    \begin{proof}[Proof of Lemma~\ref{lem:heavy-detect}]
        As mentioned,  $B^\heavy \subseteq  \{b \in \adom(B): |R_2\ltimes b| \ge \sqrt{\Lambda}\}$ always holds. It suffices to show that $\{b \in \adom(B): |R_2\ltimes b| \ge \sqrt{\Lambda}\} \subseteq B^\heavy$. Consider an arbitrary value $b \in \adom(B)$ with $|R_2\ltimes b| \ge \sqrt{\Lambda}$. In each attempt, the probability that no tuple in $R_2\ltimes b$ is sampled is $1- \frac{|R_2\ltimes b|}{|R_2|}$. Then, in $k = \frac{N\cdot \log\frac{N}{\delta}}{\sqrt{\Lambda}}$ independent attempts, the probability that no tuple in $R_2\ltimes b$ is sampled is $\left(1-\frac{|R_2\ltimes b|}{|R_2|}\right)^{\frac{N\cdot \log\frac{N}{\delta}}{\sqrt{\Lambda}}} \le \left(1-\frac{\sqrt{\Lambda}}{N}\right)^{\frac{N\cdot \log\frac{N}{\delta}}{\sqrt{\Lambda}}} \le \textsf{exp}(-\log \frac{N}{\delta}) \le \frac{\delta}{N}$. So, $b$ is not detected with probability at most $\frac{\delta}{N}$. As there are at most $N$ such values, by union bound, the probability that at least one of them is not detected is at most $\delta$.
        \end{proof}

    \begin{lemma}
    \label{lem:approximate-count-with-threshold}
    In Algorithm~\ref{alg:approximate-count-with-threshold}, the following holds for an arbitrary input instance $\R$ of input size $N$ and output size $\OUT$: 
    \begin{itemize}[leftmargin=*]
        \item  When $\tau \le \frac{1}{2} \cdot \OUT$, Algorithm~\ref{alg:approximate-count-with-threshold} returns $s\neq \perp$ with probability at least $1-\delta$;
        \item When  $s\neq \perp$ is returned, we always have $\expt[s] = \OUT$.
        \item When $s\neq \perp$ is returned, $s$ is an $\epsilon$-approximation of $\OUT$ with probability at least $1-\delta$.
    \end{itemize}
    \end{lemma}

     \begin{proof}[Proof of Lemma~\ref{lem:approximate-count-with-threshold}]
        Now, consider an arbitrary iteration of the while-loop with guess $\lambda$. We show that $\expt[s_i] = \OUT$. Each invocation of {\sc SampleMatrix} successfully returns any sample with probability 
        $\frac{\OUT}{\widehat{\OUT}_\Join}$. 
        Let $X$ be the number of times when {\sc SampleMatrix} succeeds out of $k_\lambda$ attempts. Hence, 
        $\expt[X] =\frac{k_\lambda \cdot \OUT}{\widehat{\OUT}_\Join}$ 
        and 
        $\var[X] \le \frac{k_\lambda \cdot \OUT}{\widehat{\OUT}_\Join}$. 
        Hence, 
        $\expt[s_i] = \frac{\widehat{\OUT}_\Join}{k_\lambda} \cdot \expt[X] = \OUT$. 
        Hence, if Algorithm~\ref{alg:approximate-count-with-threshold} returns an estimator $s$, we must have $\expt[s] = \OUT$. 
        By the Markov inequality, $\Pr[s_i \ge 2\OUT] \le \frac{1}{2}$. Taking the median of $k_\delta$ such estimators, 
        \[\Pr[s \le 2\OUT] \ge 1-(\frac{1}{2})^{3\log \frac{2}{\delta}} = 1-(\delta/2)^3\ge1-\delta/2. \]
        Furthermore, the variance of $s_i$ is 
        $$\var[s_i] = (\frac{1}{k_\lambda})^2 \cdot \var[Z] = (\frac{\widehat{\OUT}_\Join}{k_\lambda})^2 \cdot \var[X] \le  (\frac{\widehat{\OUT}_\Join}{k_\lambda})^2 \cdot \frac{k_\lambda \cdot \OUT}{\widehat{\OUT}_\Join} = \frac{\epsilon^2\lambda}{6} \cdot \OUT.$$ 
        By Chebyshev inequality, 
        $\Pr[|s_i - \OUT| \ge \epsilon \OUT] \le \frac{\var[s_i]}{\epsilon^2 \cdot \OUT^2} \le \frac{\lambda}{6\OUT}$.
         Taking the median of $k_\delta$ such estimators, 
        \[\Pr[|s - \OUT| \le \epsilon \OUT] \ge 1- ( \frac{\lambda}{6\OUT})^{3 \log \frac{2}{\delta}}.\]
        If $\lambda \le 4\OUT$, then $\Pr[|s - \OUT| \le \epsilon \OUT] \ge 1- (1 - \frac{1}{3})^{3 \log \frac{2}{\delta}} 
        \ge 
        1- \textsf{exp}(-\log \frac{2}{\delta}) \ge 1-\delta/2$.

        Assume $\tau \le \frac{1}{2}\OUT$. If Algorithm~\ref{alg:approximate-count-with-threshold} returns $\perp$, there must be a choice of $\lambda^* \in [\frac{\OUT}{2}, \OUT]$ tried. Let us focus on the $s$ calculated by the iteration of the while-loop with $\lambda^*$. From the analysis above, $\Pr[|s - \OUT| \le \epsilon \OUT] \ge 1-\delta/2$. As $\lambda^* \le \OUT$, we have
        $$\Pr[s\ge \lambda^*] \ge \Pr[s \ge \OUT] \ge \Pr[s \ge (1-\epsilon)\OUT] \ge 1-\delta/2\ge 1-\delta,$$ hence $s$ should be returned by triggering Line \ref{break} with probability at least $1-\delta$. If Algorithm ~\ref{alg:approximate-count-with-threshold} returns $\perp$ with probability at least $\delta$, a contradiction would occur. Thus, the first bullet is proved.

         Now, assume an estimator $s^*$ is returned, together with the guess $\lambda^*$ in this specific iteration of the while-loop. There must be $s^* \ge \lambda^*$ implied by Line \ref{break}. From the analysis above, with probability at least $1-\delta/2$, $\lambda^* \le 2\OUT$. If this holds, $\Pr[|s - \OUT| \le \epsilon \OUT] \ge 1-\delta/2$. By the union bound, $s$ is an $\epsilon$-approximation of $\OUT$ with probability at least $1-\delta$.
    \end{proof}

    \begin{lemma}
    \label{lem:intersect-estimate}
        In Algorithm~\ref{alg:intersect-estimate}, with probability at least $1-\delta$, $\left|\frac{|S'|}{|S|} - \frac{|\matrix(\R^\heavy) \cap \matrix(\R^\light)|}{\OUT^\light}\right| \le \epsilon$. 
    \end{lemma}

    \begin{proof}[Proof of Lemma~\ref{lem:intersect-estimate}]
        Let $J = \matrix(\R^\light)$ be the set of query results from the light part, and let $I = \matrix(\R^\heavy) \cap \matrix(\R^\light)$ be the intersection. Let $p = \frac{|I|}{|J|}$ be the true fraction of the intersection. Algorithm~\ref{alg:intersect-estimate} draws a set $S$ of independent uniform samples from $J$. For each sample $s_i \in S$, let $X_i$ be an indicator random variable such that $X_i = 1$ if $s_i \in I$ (i.e., $s_i \in S'$), and $X_i = 0$ otherwise. Then $X_i$ follows a Bernoulli distribution with parameter $p$, and $\expt[X_i] = p$. The algorithm returns the estimator $\hat{p} = \frac{|S'|}{|S|} = \frac{1}{|S|} \sum_{s_i \in S} X_i$. Let $M = |S| = \frac{1}{2\epsilon^2} \log \frac{2}{\delta}$ be the sample size. 
        By Hoeffding's inequality, for any $t > 0$:
            $ \Pr[|\hat{p} - p| \ge t] \le 2 \textsf{exp}(-2 M t^2). $
            Setting $t = \epsilon$, 
            $\Pr[|\hat{p} - p| \ge \epsilon] \le 2\textsf{exp}(-2 \frac{1}{2\epsilon^2} \log \frac{2}{\delta} \cdot \epsilon^2)\le \delta. $
            This implies that $|\hat{p} - p| \le \epsilon$ with probability at least $1-\delta$.
        \end{proof}

    \begin{lemma}
    \label{lem:approx-counting-matrix-with-guess}
        Algorithm~\ref{alg:approx-counting-matrix-with-guess} returns an estimator $s_\Lambda$ with $\expt[s_\Lambda]\le 2\OUT$. Furthermore, if $\Lambda \le 4 \OUT$, $s_\Lambda$ is an $\epsilon$-approximation estimator for $\OUT$ with probability at least $1-\delta$.
    \end{lemma}

    \begin{proof}[Proof of Lemma~\ref{lem:approx-counting-matrix-with-guess}]
        Notice that Line 2 succeeds or not only affects the running time, not the accuracy.
        Implied by Lemma~\ref{lem:approximate-count-with-threshold}, if $s^\heavy$ (resp., $s^\light$) is not assigned $\perp$ at Line 6 (resp., Line 7), then $\expt[s^\heavy] = \OUT^\heavy$ (resp., $\expt[s^\light] = \OUT^\light$). 
        Even if both $s^\heavy$ and $s^\light$ are not $\perp$, we have $\expt[s_\Lambda] \le 
        \expt[s^\heavy +s^\light]= 
        \OUT^\heavy + \OUT^\light \le 2 \OUT$. Analyze the other $3$ cases of $s^\heavy$ and $s^\light$ analogously, we also have $\expt[s_\Lambda] \le2\cdot \OUT$.
        
        From above, we have $\Pr[s_\Lambda > 4\OUT] < 1/2$. Now, suppose $\Lambda \le 4 \OUT$. We point out the fact that $\max\{\OUT^\heavy, \OUT^\light\} \ge \frac{\OUT}{2} \ge \frac{\Lambda}{8} \ge \frac{\epsilon'\Lambda}{8}$. Thus, at least one of $\OUT^\heavy, \OUT^\light$ is larger than $2\cdot\frac{\epsilon'\Lambda}{16}=\frac{\epsilon'\Lambda}{8}$. Implied by Lemma~\ref{lem:approximate-count-with-threshold}, at least one of $s^\heavy, s^\light$ is not $\perp$ with probability at least $1-\delta/4$. Conditioned on this, $s_\Lambda$ is returned from the first three cases:

        \item {\bf Case 1: $s^\heavy \neq \perp$ and $s^\light \neq \perp$.}
        \begin{itemize}[leftmargin=*]
        Conditioned on Line 9 succeeding, we further analyze:
        \item {\bf Case 1.1: $s^
                \heavy \ge \beta\cdot s^\light$.} 
In this subcase, $s_\Lambda = s^\heavy + s^\light(1-\beta)$.
                We know that $s^\light \cdot \beta$ is an estimator for the intersection size $I = |\matrix(\R^\heavy) \cap \matrix(\R^\light)|$. Specifically,
                \begin{align*}
                    |s^\light \beta - I| 
                    &\le |s^\light \beta - \OUT^\light \beta| + \left|\OUT^\light \beta - I\right| \\
                    &\le \epsilon' \cdot \OUT^\light \cdot \beta + \OUT^\light \cdot \epsilon' \le \epsilon'(1+\epsilon')\OUT^\light + \epsilon' \OUT^\light \le 3\epsilon' \OUT^\light.
                \end{align*}
                The total error of the estimator $s_\Lambda$ is:
                \begin{align*}
                    |s_\Lambda - \OUT| &= |(s^\heavy + s^\light - s^\light \beta) - (\OUT^\heavy + \OUT^\light - I)| \\
                    &\le |s^\heavy - \OUT^\heavy| + |s^\light - \OUT^\light| + |I - s^\light \beta| \\
                    &\le \epsilon' \OUT^\heavy + \epsilon' \OUT^\light + 3\epsilon' \OUT^\light \le 5\epsilon' \OUT=\epsilon \OUT.
                \end{align*}
                Thus, $s_\Lambda$ is an $\epsilon$-approximation of $\OUT$.
               
            \item {\bf Case 1.2: $s^\heavy < \beta \cdot s^\light$.}
                In this subcase, $s_\Lambda = s^\light$.
                The condition $s^\heavy < s^\light \beta$ implies that the estimated intersection is larger than the estimated heavy set. Since the true intersection $I$ is a subset of $\matrix(\R^\heavy)$, we have $I \le \OUT^\heavy$.
                Using the bounds from Case 1.1, we have:
                \[ (1-\epsilon')\OUT^\heavy \le s^\heavy < s^\light \beta \le I + 3\epsilon' \OUT^\light. \]
                This implies $\OUT^\heavy - I < \epsilon' \OUT^\heavy + 3\epsilon' \OUT^\light$.
                The true output size is $\OUT = \OUT^\light + (\OUT^\heavy - I)$. Substituting the bound above:
                \[ \OUT^\light \le \OUT < \OUT^\light + \epsilon' \OUT^\heavy + 3\epsilon' \OUT^\light \le \OUT^\light + 4\epsilon' \OUT. \]
                Since $s^\light$ is an $\epsilon'$-approximation of $\OUT^\light$, we have:
                \[ s^\light \le (1+\epsilon')\OUT^\light \le (1+\epsilon')\OUT\le (1+\epsilon)\OUT,  \textrm{ and }\]
                \[ s^\light \ge (1-\epsilon')\OUT^\light \ge (1-\epsilon')(\OUT - 4\epsilon' \OUT) \ge (1-5\epsilon')\OUT=(1-\epsilon) \OUT. \]
                Thus, $s_\Lambda = s^\light$ is an $\epsilon$-approximation of $\OUT$.
        \end{itemize}
        
        \item {\bf Case 2: $s^\heavy \neq \perp$ and $s^\light = \perp$.} 
        In this case, $s_\Lambda = s^\heavy$. Following Lemma~\ref{lem:approximate-count-with-threshold}, with probability at least $1-\delta/4$, $\OUT^\light < \frac{\epsilon'\Lambda}{8}$, and $s^\heavy$ is an $\epsilon'$-approximation of $\OUT^\heavy$. Moreover, $\OUT^\heavy \ge \frac{\OUT}{2} \ge \frac{\Lambda}{8}$.
         We work out an upper bound of $\OUT$ in terms of $\OUT^\heavy$:
        $$\OUT \le \OUT^\heavy + \OUT^\light \le  \OUT^\heavy + \frac{\epsilon'\Lambda}{8} \le  \OUT^\heavy + \frac{\epsilon'}{2} \cdot \OUT,$$
        i.e., $\OUT (1-\frac{\epsilon'}{2})\le \OUT^\heavy$.
        As $s^\heavy$ is an $\epsilon'$-approximation of $\OUT^\heavy$, $
        |s^\heavy -\OUT^\heavy| \le \epsilon'\OUT^\heavy$. Putting together with $\OUT^\heavy \le \OUT$ and $\OUT^\heavy \ge \OUT (1-\frac{\epsilon'}{2})$, we get:
        $$(1-\epsilon)\OUT\le(1-\frac{3\epsilon'}{2})\OUT \le  (1-\epsilon')(1-\frac{\epsilon'}{2})\OUT \le s_\Lambda \le (1+\epsilon')\OUT\le (1+\epsilon)\OUT.$$
        \item {\bf Case 3: $s^\heavy = \perp$ and $s^\light \neq \perp$} Similar to Case 2.
    
    Thus, if $\Lambda\le 4\OUT$, with probability at least $1-\delta$, $s_\Lambda$ is an $\epsilon$-approximation for $\OUT$.
    \end{proof}

      \begin{lemma}
        \label{lem:approx-counting-matrix}
        Algorithm~\ref{alg:approx-counting-matrix} returns an $\epsilon$-approximation of $\OUT$ with probability at least $1-\delta$.
    \end{lemma}

    \begin{proof}[Proof of Lemma~\ref{lem:approx-counting-matrix}]
         Consider the invocation of Algorithm~\ref{alg:approx-counting-matrix-with-guess} at Line 4. From Lemma~\ref{lem:approx-counting-matrix-with-guess}, $\expt[s_i]\le 2\OUT$. By the Markov inequality, $\Pr[s_i > 4\OUT] \le \frac{1}{2}$.
        Taking the median, $\Pr[s \le 4\OUT] \ge 1- (\frac{1}{2})^{\log \frac{2}{\delta}} \ge 1- \frac{\delta}{2}$. Let $s^*$ be the final estimator returned with the guess $\Lambda^*$. We know that $s^* \ge \Lambda^*$. Hence, $\Lambda^* \le 4 \OUT$ holds with probability at least $1-\frac{\delta}{2}$. Following Lemma~\ref{lem:approx-counting-matrix-with-guess}, when $\Lambda^* \le 4 \OUT$, each $s_i$ is an $\epsilon$-approximation of $\OUT$ with probability at least $1-\frac{\delta}{2}$. Thus, $s^*$ is an $\epsilon$-approximation of $\OUT$ with probability at least $1-\frac{\delta}{2}$. By the union bound, $s^*$ is an $\epsilon$-approximation of $\OUT$ with probability at least $1-\delta$.
    \end{proof}
    
    We next analyze the time complexity 
    of our approximate counting algorithm: 
    
    \begin{lemma}
        \label{lem:count-matrix-up}
        Algorithm~\ref{alg:approx-counting-matrix} runs in $O\left(\frac{N \cdot \log(N/\delta) \cdot \log(1/\delta)}{\epsilon^{2.5} \sqrt{\OUT}}\right)$ time with probability at least $1-\delta$.
    \end{lemma}
    
     \begin{proof}[Proof of Lemma~\ref{lem:count-matrix-up}]
    We analyze the total time complexity by summing the cost of \textsc{ApproxCountMatrixWithGuess} over the sequence of guesses $\Lambda = N^2, N^2/2, \dots, 1$. The algorithm terminates with high probability when $\Lambda \le \OUT$. Thus, we only consider $\Lambda \gtrsim\footnote{greater than or approximately equal to.} \OUT$.
    
    For a fixed $\Lambda$, the procedure \textsc{ApproxCountWithThreshold} iterates guesses $\lambda$ from $N^2$ down to a threshold $\tau = \frac{\epsilon' \Lambda}{16} = \Theta(\epsilon \Lambda)$. The cost for a specific $\lambda$ is dominated by the sampling loop, which performs $k_\lambda = O(\frac{\widehat{\OUT}_\Join}{\epsilon^2 \lambda})$ iterations. The expected cost per iteration is $O(1 + \frac{N\cdot \sqrt{\OUT}}{\widehat{\OUT}_\Join})$. Thus, the total cost for a fixed $\lambda$ is:
    $T(\lambda) \approx \frac{\widehat{\OUT}_\Join}{\epsilon^2 \lambda} \left(1 + \frac{N\cdot\sqrt{\OUT}}{\widehat{\OUT}_\Join}\right) = \frac{1}{\epsilon^2 \lambda} (\widehat{\OUT}_\Join + N\sqrt{\OUT}).$
    For the light part, $\widehat{\OUT}_\Join^\light = N\sqrt{\Lambda}$. For the heavy part, $\widehat{\OUT}_\Join^\heavy \le N\sqrt{\OUT^\heavy} \le N\sqrt{\Lambda}$ (since $\OUT^\heavy \le \OUT \le \Lambda$ in the relevant range). Thus, $T(\lambda) \le O(\frac{N\sqrt{\Lambda}}{\epsilon^2 \lambda})$.
    
    The inner loop stops when $\lambda \approx \OUT$ (if $\OUT \ge \tau$) or reaches $\tau$ (if $\OUT < \tau$). Let $\lambda_{stop} = \max(\Theta(\OUT), \tau)$. The total cost for a fixed $\Lambda$ is dominated by the smallest $\lambda$, i.e., $\lambda_{stop}$:
    $ \text{Cost}(\Lambda) = \sum_{\lambda \ge \lambda_{stop}} T(\lambda) = O\left( \frac{N\sqrt{\Lambda}}{\epsilon^2 \lambda_{stop}} \right). $
    
    We now sum $\text{Cost}(\Lambda)$ over the geometric sequence of $\Lambda$ from $N^2$ down to $\OUT$. We split the sum into two ranges based on whether $\tau \le \OUT$ (i.e., $\epsilon \Lambda \lesssim \OUT$) or $\tau > \OUT$.
    
    \begin{itemize}[leftmargin=*]
        \item \textbf{Range 1: $\OUT \le \Lambda \le \frac{\OUT}{\epsilon}$.}
        Here, $\tau \approx \epsilon \Lambda \le \OUT$, so the inner loop stops at $\lambda_{stop} \approx \OUT$.
        $\text{Cost}(\Lambda) \approx \frac{N\sqrt{\Lambda}}{\epsilon^2 \OUT}. $
        Summing over $\Lambda = 2^k \OUT$ for $k=0 \dots \log(1/\epsilon)$:
        $\sum_{\Lambda} \text{Cost}(\Lambda) \approx \frac{N}{\epsilon^2 \OUT} \sum_{k} \sqrt{2^k \OUT} = \frac{N}{\epsilon^2 \sqrt{\OUT}} \sum_{k} (\sqrt{2})^k. $
        The sum is dominated by the largest term ($\Lambda \approx \OUT/\epsilon$):
        $ \text{s}_1 \approx \frac{N}{\epsilon^2 \sqrt{\OUT}} \cdot \sqrt{\frac{1}{\epsilon}} = \frac{N}{\epsilon^{2.5} \sqrt{\OUT}}. $
        
        \item \textbf{Range 2: $\Lambda > \frac{\OUT}{\epsilon}$.}
        Here, $\tau \approx \epsilon \Lambda > \OUT$, so the inner loop stops at $\lambda_{stop} = \tau \approx \epsilon \Lambda$.
        $ \text{Cost}(\Lambda) \approx \frac{N\sqrt{\Lambda}}{\epsilon^2 (\epsilon \Lambda)} = \frac{N}{\epsilon^3 \sqrt{\Lambda}}.  $
        Summing over $\Lambda = 2^k (\OUT/\epsilon)$:
         $ \sum_{\Lambda} \text{Cost}(\Lambda) \approx \frac{N}{\epsilon^3} \sum_{\Lambda} \frac{1}{\sqrt{\Lambda}}.  $
        The sum is dominated by the smallest term ($\Lambda \approx \OUT/\epsilon$):
         $ \text{s}_2 \approx \frac{N}{\epsilon^3 \sqrt{\OUT/\epsilon}} = \frac{N}{\epsilon^{2.5} \sqrt{\OUT}}.  $
    \end{itemize}
    
    Combining both ranges, the total time complexity is $O(\frac{N}{\epsilon^{2.5} \sqrt{\OUT}})$, multiplied by the logarithmic factors from probability amplification.
    \end{proof}


    \begin{theorem}
    \label{thm:count-matrix-up}
        For $\matrix$, given parameters $0<\epsilon, \delta <1$, there is an algorithm that can take as input an arbitrary input instance $\R$, and output an $\epsilon$-approximation of $|\matrix(\R)|$ in $O\left(\frac{N \cdot \log(N/\delta) \cdot \log(1/\delta)}{\epsilon^{2.5} \sqrt{\OUT}}\right)$ time with probability at least $1-\delta$.
    \end{theorem}


    \section{Lower Bounds for Matrix Query}
    \label{sec:lb-matrix}
 
    In this section, we establish lower bounds for approximate counting and uniform sampling matrix query. Our lower bound for approximate counting is a reduction from the set disjointness problem in the two-party communication model. Given the universe $[m]$, Alice is given a set $S_A \subseteq [m]$ of $\frac{m}{4}$ elements and Bob is given a set $S_B \subseteq [m]$ of $\frac{m}{4}$ elements. The goal is to determine if $S_A \cap S_B = \emptyset$. It is known that any randomized communication protocol that solves this problem with constant success probability requires $\Omega(m)$ bits of communication~\cite{kalyanasundaram1992probabilistic, bar2002information,razborov1990distributional}.
    
    \begin{lemma}\label{lmm:count-matrix-lb}
        For $\matrix$, given any nonnegative integers $N, \OUT, n$, any algorithm that takes as input an arbitrary instance $\R$ of input size $N$, output size $\OUT$, and active domain size $n$ for $B$, and returns an $O(1)$ approximation of $\OUT$ with high probability, requires at least $\Omega(\min\{n, \frac{N}{\sqrt{\OUT}}\})$ time.
    \end{lemma}
    \begin{proof}[Proof of Lemma~\ref{lmm:count-matrix-lb}]
    Given an input instance $(S_A,S_B)$ of the set disjointness problem, Alice and Bob construct an instance $\R$ of $\matrix$ based on their private sets $S_A$ and $S_B$ as follows. The domain of attribute $B$ is the universe $[m]$. The active domain of attribute $B$ is $S_A \cup S_B$. The active domains for attributes $A$ and $C$ are two disjoint sets, each of size $\sqrt{K}$ for some $K>1$. Alice defines relation $R_1(A,B)$ as $R_1 = \adom(A) \times S_A$. Bob defines relation $R_2(B,C)$ as $R_2 = S_B \times \adom(C)$.
    The input size of $\R$ is $N = |R_1| + |R_2| = \sqrt{K}|S_A| + \sqrt{K}|S_B| = \sqrt{K} \cdot m/2$.
    The output size of $\R$ is $\OUT = |\pi_{A,C}(R_1 \Join R_2)|$ depends on the intersection of $S_A$ and $S_B$. If $S_A \cap S_B = \emptyset$, then $\OUT = 0$. If $S_A \cap S_B \neq \emptyset$, then $\OUT = |\adom(A)| \cdot |\adom(C)| = K$.

    Any algorithm $\mathcal{A}$ for approximate counting $\R$ can be used to solve the set disjointness instance $(S_A,S_B)$ since it distinguishes between the $\OUT=0$ case and the $\OUT=K$ case. Each query made by $\mathcal{A}$ will be answered by Alice and Bob communicating a small number of bits. For any query $\mathcal{A}$ makes (e.g., a tuple check of $(a,b) \in R_1$), Alice can answer it locally by checking if $b \in S_A$ and communicating the 1-bit answer to Bob. Other primitives are handled analogously. For queries on $R_1$ involving $b$ (degree or neighbor access), Alice uses the fixed set $\adom(A)$ conditioned on $b \in S_A$. For queries involving $a$ (degree or neighbor access), Alice uses the private set $S_A$ (e.g., the degree of any $a \in \adom(A)$ is $|S_A|$, and the neighbors are the elements of $S_A$). Finally, for a sample query, Alice draws uniformly from $\adom(A) \times S_A$. Bob handles queries on $R_2$ symmetrically.
    If $\mathcal{A}$ provides an $O(1)$ approximation of $\OUT$ within $o(\min\{n,\frac{N}{\sqrt{\OUT}}\})=o(n)$ queries, the set disjointness problem is solved within $o(n)=o(m)$ communication cost, leading to a contradiction.
    \end{proof}

    With this lemma, we can derive the following lower bound that is only parameterized by $N$ and $\OUT$, as it is always feasible to construct an instance with $n = \Theta(\frac{N}{\sqrt{\OUT}})$: 

    \begin{theorem}
        \label{thm:count-matrix-lb}
        For $\matrix$, given any nonnegative integers $N, \OUT$, any randomized algorithm that takes as input an arbitrary instance $\R$ of input size $N$ and output size $\OUT$, and outputs an $O(1)$ approximation of $\OUT$ with high probability, requires at least $\Omega(\frac{N}{\sqrt{\OUT}})$ time.
    \end{theorem}

     \begin{proof}[Proof of Theorem~\ref{thm:count-matrix-lb}] For any given $N, \OUT$, we distinguish the following two cases:
        
        If $\OUT \ge 1$, there is a hard instance for $\matrix$ with $n=\Theta(\frac{N}{2\sqrt{\OUT}})$ as follows. The domain sizes of attribute $A,B,C$ are $\sqrt{\OUT},\frac{N}{2\sqrt{\OUT}},\sqrt{\OUT}$ respectively. $R_1$ is the Cartesian product of $\dom(A)$ and $\dom(B)$. $R_2$ is the Cartesian product of $\dom(B)$ and $\dom(C)$. So, any algorithm needs $\Omega(\min\{n, \frac{N}{\sqrt{ \OUT}}\})=\Omega(\frac{N}{\sqrt{\OUT}})$ time. 
            
        If $\OUT = 0$, there is a hard instance for $\matrix$ with $n=N$ as follows. The domain sizes of attribute $A,B,C$ are $1,N,1$ respectively. Relation $R_1$ is the Cartesian product of $\dom(A)$ and any $\frac{N}{2}$ values in $\dom(B)$. Relation $R_2$ is the Cartesian product of the remaining $\frac{N}{2}$ values in $\dom(B)$ and $\dom(C)$.
        So, any algorithm needs to spend $\Omega(\min\{n, \frac{N}{\sqrt{ \OUT}}\})=\Omega(N)$ time. 
        
        Combining these two cases yields the desired lower bound. 
        \end{proof}

    This lower bound is derived from communication complexity, which imposes an information-theoretic limit independent of the specific computational operations used. Consequently, this applies to \emph{any} algorithm, including those based on fast matrix multiplication, effectively ruling out the possibility of breaking the $\Omega(\frac{N}{\sqrt{\OUT}})$ barrier using algebraic techniques.

        Following Proposition~\ref{prop:reduction}, Theorem~\ref{thm:count-matrix-lb} automatically implies a lower bound on uniform sampling over matrix query (for any algorithms), such that no index can be preprocessed in $O\left(\frac{N}{\OUT^{\frac{1}{2}+\gamma}}\right)$ time, while generating each (independent) uniform sample in $O\left(\frac{N}{\OUT^{\frac{1}{2}+\gamma}}\right)$ time, for any constant $\gamma >0$. But this is not sufficient to show the optimality of our sampling algorithm, which uses $O(N)$ preprocessing time and generates samples in $O(\frac{N}{\sqrt{\OUT}})$ time. Below, we give a stronger lower bound for combinatorial algorithms following \cite{deng2023join}:
         
    \begin{theorem}
    \label{thm:sample-matrix-combinatorial-lb}
        For $\matrix$, assuming the combinatorial hardness of Boolean matrix multiplication\footnote{Given two Boolean matrices of sizes $n \times n$ and $n \times n$, any combinatorial algorithm for computing their multiplication requires $\Omega(n^3)$ time.}, given any nonnegative integers $N, \OUT$, there is no combinatorial algorithm that preprocesses an instance $\R$ of input size $N$ and output size $\OUT$ in $O(N \cdot \OUT^{\frac{1}{2} -\gamma})$ time and generates a uniform sample from $\matrix(\R)$ in $O\left(\frac{N}{\OUT^{\frac{1}{2} + \gamma}}\right)$ expected time, for any constant $\gamma > 0$.
    \end{theorem}

       \begin{proof}[Proof of Theorem~\ref{thm:sample-matrix-combinatorial-lb}]
        We prove this by reduction from the problem of listing all query results. It is established that any combinatorial algorithm for evaluating the $\matrix$ requires $\Omega(N\sqrt{\OUT})$ time~\cite{amossen2009faster}. Suppose, for the sake of contradiction, that there exists a sampling algorithm $\mathcal{A}$ with preprocessing time $O(N \cdot \OUT^{\frac{1}{2} -\gamma})$ and expected sampling time $T = O\left(\frac{N}{\OUT^{1/2+\gamma}}\right)$ for some constant $\gamma > 0$. We construct an algorithm to compute the full join result $\matrix(\R)$ without prior knowledge of $\OUT$, using the estimation technique from~\cite{deng2023join}.
        
        We repeatedly invoke $\mathcal{A}$ to generate samples and store them in a dictionary (e.g., a hash table) to track distinct tuples. We maintain a counter $c$ for the number of \emph{consecutive} samples that are already in the dictionary. We stop this process as soon as $c$ reaches a threshold $\Delta = \Theta(\log N)$. Let $k$ be the number of distinct tuples collected when we stop. We set an estimate $\Lambda = 2k$. As shown in~\cite{deng2023join}, this stopping condition implies that with high probability, we have collected at least half of the total results, i.e., $k \ge \OUT/2$. Consequently, $\Lambda \in [\OUT, 2\OUT]$ with high probability. The number of samples drawn in this step is $\O(\OUT)$.
    
        Using the estimate $\Lambda$, we continue invoking $\mathcal{A}$ for a total of $\Theta(\Lambda \log N)$ trials. Since $\Lambda \ge \OUT$, this satisfies the requirement for the Coupon Collector's Problem ($\Theta(\OUT \log \OUT)=\Theta(\OUT \log N)$ samples) to ensure that all distinct tuples in $\matrix(\R)$ are collected with high probability.
    
        The total expected number of samples drawn across both steps is $O(\OUT \log N)$. The total running time to enumerate the full query result is:
        \begin{align*}
            O(N \cdot \OUT^{\frac{1}{2} -\gamma}) + O\left(\OUT \log N \cdot \frac{N}{\OUT^{1/2+\gamma}}\right) = O\left(N \cdot \OUT^{1/2-\gamma} \log N\right).
        \end{align*}
        For any $\OUT \ge \log N$, this runtime is strictly faster than $O(N\sqrt{\OUT})$, contradicting the combinatorial hardness of matrix multiplication. 
    \end{proof}

     \section{Star Query}
    \label{appendix:star}
    
    In this section, we generalize our results from the matrix query to the star query. A star query with $k$ relations is defined as $\star = (\V, \E, \y)$, where $\V = \{A_1, \dots, A_k, B\}$, $\E = \{e_1, \dots, e_k\}$ with $e_i = \{A_i, B\}$, and the output attributes are $\y = \{A_1, \dots, A_k\}$. The query can be written as:
    \[
        \star = \pi_{A_1, \dots, A_k} R_1(A_1, B) \bowtie R_2(A_2, B) \bowtie \dots \bowtie R_k(A_k, B).
    \]
    Note that the matrix query $\matrix$ corresponds to the case where $k=2$ (with $A_1=A, A_2=C$).
    
    \subsection{Uniform Sampling over Star Query}
    
    The sampling framework follows the same structure as Algorithm~\ref{alg:sample-matrix}: (Step 1) Sample a full join result $\vec{t} = (a_1, \dots, a_k, b)$ from $R_1 \bowtie \dots \bowtie R_k$ uniformly at random; (Step 2) Accept $\vec{a} = (a_1, \dots, a_k)$ with probability $\frac{1}{\deg(\vec{a})}$, where $\deg(\vec{a})$ is the number of $B$-values connecting $a_1, \dots, a_k$.
    
    \paragraph{Step 1: sample full join result}
    We preprocess the instance to compute the weight of each $b \in \adom(B)$ as $W_b = \prod_{i=1}^k |R_i \ltimes b|$. Let $W = \sum_{b \in \adom(B)} W_b$.
    To sample a full join tuple:
    \begin{itemize}[leftmargin=*]
        \item Sample $b \in \adom(B)$ with probability $W_b/W$.
        \item For each $i \in [k]$, sample $a_i$ uniformly from $\pi_{A_i}(R_i \ltimes b)$.
    \end{itemize}
    The tuple $(a_1, \dots, a_k, b)$ is returned with probability at least $1/W = 1/\OUT_\Join$.
    
    \paragraph{Step 2: acceptance check}
    For a candidate tuple $\vec{a} = (a_1, \dots, a_k)$, the degree is defined as:
    \[ \deg(\vec{a}) = \left| \bigcap_{i=1}^k \pi_B(R_i \ltimes a_i) \right|. \]
    To accept with probability $\frac{1}{\deg(\vec{a})}$ without fully computing the intersection, we generalize Algorithm~\ref{alg:accept}. We identify the relation with the smallest degree for the given tuple. Let $j = \arg\min_{i \in [k]} |R_i \ltimes a_i|$. We set $S = \pi_B(R_j \ltimes a_j)$.
    We then sample $b'$ from $S$ (excluding the witness $b$) and check if $b'$ connects all other $a_i$ (i.e., $\forall i \neq j, (a_i, b') \in R_i$). The number of failures $F$ before finding a valid witness follows a negative hypergeometric distribution. We accept if $\randint(1, |S|) \le F+1$.
    
    \begin{algorithm}[t]
        \caption{\textsc{SampleStar}$(\R, k)$}
        \label{alg:sample-star}
        \KwIn{An instance $\R$ for $\star$, with precomputed weights $W_b$.}
        \KwOut{A query result of $\star(\R)$.}
            $b \gets$ sample from $\adom(B)$ with prob. $W_b/W$\;
            \lFor{$i \in [k]$}{$a_i \gets$ uniform sample from $\pi_{A_i}(R_i \ltimes b)$}
            $j \gets \arg\min_{i \in [k]} |R_i \ltimes a_i|$\; 
            $S \gets \pi_B(R_j \ltimes a_j)$\;
            $F \gets 0$\;
            \While{at least one element in $S$ is not visited}{
                $b' \gets$ random sample from non-visited in $S$\;
                \lIf{$b'=b$}{\Continue}
                \lIf{$\exists i \neq j, (a_i, b') \notin R_i$}{$F \gets F+1$}
                \lElse{\Break}
            }
            \lIf{$\randint(1, |S|) \le F+1$}{\Return $(a_1, \dots, a_k)$}
        
    \end{algorithm}
    
    \paragraph{Analysis}
    The correctness follows the same logic as Section~\ref{sec:sampling-matrix-analysis}. For the time complexity, the expected cost per sample is proportional to:
    \[
        1 + \frac{1}{\OUT_\Join} \sum_{\vec{a} \in \star(\R)} \min_{i \in [k]} |R_i \ltimes a_i|.
    \]
    Using the generalized AGM bound logic for star queries, it is known that $\sum_{\vec{a}} \min_i |R_i \ltimes a_i| \le N \cdot \OUT^{1 - 1/k}$.
    Thus, the expected runtime per successful sample is:
    \[
        O\left( \frac{\OUT_\Join}{\OUT} \cdot \left( 1 + \frac{N \cdot \OUT^{1 - 1/k}}{\OUT_\Join} \right) \right) = O\left( \frac{\OUT_\Join}{\OUT} + \frac{N}{\OUT^{1/k}} \right).
    \]
    Since $\OUT_\Join \le N \cdot \OUT^{1-1/k}$ for star queries, the term simplifies to $O(N / \OUT^{1/k})$.
    
    \begin{theorem}\label{thm:star-sample-ub}
        For the star query $\star$, there is an algorithm that preprocesses an instance $\R$ of size $N$ in $O(N)$ time and generates a uniform sample from $\star(\R)$ in $O\left(\frac{N}{\OUT^{1/k}}\right)$ expected time. 
    \end{theorem}
    
    \subsection{Approximate Counting Star Query}
    
    The approximate counting algorithm for $\star$ follows the hybrid strategy in Section~\ref{sec:count-matrix}. We partition the domain of $B$ into heavy and light values based on a threshold $\tau$.
    For the star query, the optimal threshold logic shifts from $\sqrt{\OUT}$ to $\OUT^{1/k}$.
    
    \begin{itemize}[leftmargin=*]
        \item \textbf{Heavy Detect:} We identify $B^\heavy = \{b \in \adom(B) \mid \prod_{i=1}^k |R_i \ltimes b| \ge \tau \}$. This can be done by sampling or by iterating if we define heaviness based on the max degree in any single relation.
        \item \textbf{Hybrid Estimation:}
        \begin{itemize}[leftmargin=*]
            \item For the heavy part, we use the weighted sampling (Algorithm~\ref{alg:sample-star} adapted) which is efficient because the number of heavy $B$ values is small.
            \item For the light part, we use a rejection-based sampler (analogous to \textsc{SampleMatrix-L}) where we sample $a_1 \in R_1$, then $b \in R_1 \ltimes a_1$, then $a_2 \dots a_k$. We accept based on the ratio of the tuple weight to the maximum possible light weight.
        \end{itemize}
    \end{itemize}
    
    By setting the threshold parameters appropriately and applying the reduction from counting to sampling, we achieve the following bound:
    
    \begin{theorem}\label{thm:star-count-ub}
        For the star query $\star$ and a constant $0 < \epsilon < 1$, there exists an algorithm that returns an $\epsilon$-approximation of $\OUT$ with constant probability in $O\left(\frac{N}{\OUT^{1/k}}\right)$ time.
    \end{theorem}
    
    \subsection{Lower Bounds for Approximate Counting Star Query}
    
    We establish the lower bound for the star query $\star$ by a reduction from the $k$-party set disjointness problem in the communication complexity model. In this problem, there are $k$ players, and each player $i$ holds a private set $S_i \subseteq [m]$ of size $\frac{m}{2k}$. The goal is to determine if $\bigcap_{i=1}^k S_i = \emptyset$. It is known that any randomized communication protocol that solves this problem with constant success probability requires $\Omega(m)$ bits of communication.
    
    \begin{lemma}\label{lmm:count-star-lb}
        For $\star$, given any $N, \OUT, m \in \mathbb{Z}^+$, any randomized algorithm that takes as input an arbitrary instance $\R$ of input size $N$, output size $\OUT$, and active domain size $m$ of attribute $B$, and outputs any multiplicative-factor approximation of $\OUT$ with high probability requires at least $\Omega(\min\{m, \frac{N}{\OUT^{1/k}}\})$ time.
    \end{lemma}
    
    \begin{proof}
        Given an input instance $(S_1, \dots, S_k)$ of the $k$-party set disjointness problem where each $|S_i| = \frac{m}{2k}$, the $k$ players construct an instance $\R$ of $\star$ based on their private sets as follows.
        
        The domain of attribute $B$ is the universe $[m]$. The active domain of $B$ is $\bigcup_{i=1}^k S_i$. For each $i \in [k]$, the active domain of attribute $A_i$, denoted $\adom(A_i)$, is a set of size $L = K^{1/k}$ for some parameter $K \ge 1$, and all $\adom(A_i)$ sets are mutually disjoint.
        Player $i$ defines the relation $R_i(A_i, B)$ as the Cartesian product of $\adom(A_i)$ and their private set $S_i$:
        $R_i = \adom(A_i) \times S_i$.
        The input size of $\R$ is:
        \[ N = \sum_{i=1}^k |R_i| = \sum_{i=1}^k |\adom(A_i)| \cdot |S_i| = \sum_{i=1}^k K^{1/k} \cdot \frac{m}{2k} = K^{1/k} \cdot \frac{m}{2}. \]
        
        The output size $\OUT = |\star(\R)|$ depends on the intersection of the sets $S_i$:
        \begin{itemize}[leftmargin=*]
            \item If $\bigcap_{i=1}^k S_i = \emptyset$, then the full join is empty, so $\OUT = 0$.
            \item If $\bigcap_{i=1}^k S_i \neq \emptyset$, then for every $b \in \bigcap_{i=1}^k S_i$, any combination of $(a_1, \dots, a_k) \in \adom(A_1) \times \dots \times \adom(A_k)$ is a valid result. Thus, $\OUT = \prod_{i=1}^k |\adom(A_i)| = (K^{1/k})^k = K$.
        \end{itemize}
        
        Any algorithm $\mathcal{A}$ for approximate counting $\R$ can be used to solve the set disjointness instance $(S_1, \dots, S_k)$ since it distinguishes between the $\OUT=0$ case and the $\OUT=K$ case. Each query made by $\mathcal{A}$ will be answered by the players communicating a small number of bits. For example, for a tuple check $(a, b) \in R_i$, player $i$ checks if $b \in S_i$ locally (since $a \in \adom(A_i)$ is fixed and known).
        
        If $\mathcal{A}$ provides a multiplicative-factor approximation of $\OUT$ within $o(\min\{m, \frac{N}{\OUT^{1/k}}\})$ queries, the players could solve the set disjointness problem with $o(m)$ communication, which contradicts the lower bound. Note that in the case $\OUT=K$, we have $m = \frac{2N}{K^{1/k}} = \frac{2N}{\OUT^{1/k}}$.
    \end{proof}
    
    With this lemma, we can derive the following lower bound that is only parameterized by $N$ and $\OUT$, as it is always feasible to construct an instance with $m = \Theta(\frac{N}{\OUT^{1/k}})$.
    
    \begin{theorem}\label{thm:star-count-lb}
        For $\star$, given any $N, \OUT \in \mathbb{Z}^+$, any randomized algorithm that takes as input an arbitrary instance $\R$ of input size $N$ and output size $\OUT$, and outputs any multiplicative-factor approximation of $\OUT$ with high probability requires at least $\Omega\left(\frac{N}{\OUT^{1/k}}\right)$ time.
    \end{theorem}
    
    \begin{proof}
        Consider any specific parameters $N, \OUT \in \mathbb{Z}^+$. We distinguish the following two cases:
        
        \begin{itemize}[leftmargin=*]
            \item If $\OUT \ge 1$, there is a hard instance for $\star$ with $m = \Theta(\frac{N}{\OUT^{1/k}})$ as follows. The domain size of attribute $A_i$ is set to $L = \OUT^{1/k}$. The domain size of $B$ is set to $m = \frac{2N}{L} = \frac{2N}{\OUT^{1/k}}$. Each relation $R_i$ is constructed as in Lemma~\ref{lmm:count-star-lb}. Any algorithm needs $\Omega(\min\{m, \frac{N}{\OUT^{1/k}}\}) = \Omega(\frac{N}{\OUT^{1/k}})$ time.
            
            \item If $\OUT = 0$, there is a hard instance for $\star$ with $m = \Theta(N)$ as follows. The domain size of attribute $A_i$ is set to $L=1$. The domain size of $B$ is $m = N$. Each relation $R_i$ consists of the single value in $\adom(A_i)$ paired with a subset of $B$ of size roughly $N/k$. Any algorithm needs $\Omega(\min\{m, \frac{N}{1}\}) = \Omega(N)$ time.
        \end{itemize}
        
        Combining these two cases yields the desired lower bound $\Omega\left(\frac{N}{\OUT^{1/k}}\right)$.
    \end{proof}

    \subsection{Lower Bounds for Uniform Sampling Star Query}

    Similar to the matrix query, we can establish a lower bound for combinatorial algorithms by reducing from the problem of listing all query results.

    \begin{theorem}
    \label{thm:star-sample-lb}
        For the star query $\star$, no combinatorial algorithm can preprocess an instance $\R$ of input size $N$ and output size $\OUT$ in $O(N \cdot \OUT^{\frac{1}{k} -\gamma})$ time and generate a uniform sample from $\star(\R)$ in $O\left(\frac{N}{\OUT^{\frac{1}{k} + \gamma}}\right)$ expected time, for any arbitrary small constant $\gamma > 0$, assuming the combinatorial hardness of listing star join results\footnote{It is known that any combinatorial algorithm for listing all results of a star query requires $\Omega(N \cdot \OUT^{1-1/k})$ time.}, under the property testing model.
    \end{theorem}
    
    \begin{proof}
        We prove this by reduction from the problem of listing all query results. Suppose, for the sake of contradiction, that there exists a sampling algorithm $\mathcal{A}$ with preprocessing time $O(N \cdot \OUT^{\frac{1}{k} -\gamma})$ and expected sampling time $T = O\left(\frac{N}{\OUT^{1/k+\gamma}}\right)$ for some constant $\gamma > 0$. We construct an algorithm to compute the full query result $\star(\R)$ without prior knowledge of $\OUT$, following the strategy in Theorem~\ref{thm:sample-matrix-combinatorial-lb}.
    
        \paragraph{Step 1: Estimate $\OUT$}
        We repeatedly invoke $\mathcal{A}$ to generate samples and track collisions. As described in Theorem~\ref{thm:sample-matrix-combinatorial-lb}, by stopping when a sufficient number of collisions occur, we can obtain an estimate $\Lambda$ such that $\Lambda \in [\OUT, 2\OUT]$ with high probability. This step requires $O(\OUT)$ samples.
    
        \paragraph{Step 2: Coupon Collection}
        Using the estimate $\Lambda$, we invoke $\mathcal{A}$ for $\Theta(\Lambda \log N)$ trials. Since $\Lambda \ge \OUT$ and $\log \OUT = O(\log N)$, this satisfies the requirement for the Coupon Collector's Problem to ensure that all distinct tuples in $\star(\R)$ are collected with high probability.
    
        \paragraph{Complexity Analysis}
        The total number of samples drawn is $O(\OUT \log N)$. The total running time to enumerate the full query result is:
        \begin{align*}
            O(N \cdot \OUT^{\frac{1}{k} -\gamma}) + O\left(\OUT \log N \cdot \frac{N}{\OUT^{1/k+\gamma}}\right) 
            &= O\left(N \cdot \OUT^{1-\frac{1}{k}-\gamma} \log N\right).
        \end{align*}
        For sufficiently large $\OUT$, this runtime is strictly faster than the combinatorial lower bound of $\Omega(N \cdot \OUT^{1-1/k})$ for listing star query results, leading to a contradiction.
    \end{proof}

 \section{Chain Query}\label{appendix:chain}
        \subsection{Approximate Counting Chain Query}
        \label{sec:count-chain}
        
        For $\chain = \pi_{A_1, A_{k+1}} (R_1(A_1, A_2) \bowtie \dots \bowtie R_k(A_k, A_{k+1}))$, we present an algorithm that runs in $\O(N)$ time, based on the $k$-Minimum Values (KMV) summary~\cite{cohen2008tighter,cohen2007summarizing}. The core idea is to estimate the number of distinct reachable values in $A_{k+1}$ for each value in $A_1$. Let $\deg(u) = |\{v \in \adom(A_{k+1}) \mid u \rightsquigarrow v\}|$ be the number of $A_{k+1}$ values reachable from $u \in \adom(A_1)$. The total output size is $\OUT = \sum_{u \in \adom(A_1)} \deg(u)$. We can estimate each $\deg(u)$ efficiently using bottom-up dynamic programming combined with min-wise hashing.
        
        \paragraph{Algorithm}
        We assign a random hash function $h: \dom(A_{k+1}) \to [0,1]$. For any node $u$ in the join graph, we wish to maintain the $K$-th smallest hash value among all $v \in \adom(A_{k+1})$ reachable from $u$. Let $H(u)$ denote the set of hash values of all $A_{k+1}$ values reachable from $u$. We maintain a summary $S(u)$ containing the $K$ smallest values in $H(u)$, where $K = O(1/\epsilon^2)$.
        
        The algorithm proceeds in a backward pass from $R_k$ to $R_1$:
        \begin{itemize}[leftmargin=*]
            \item \textbf{Initialization:} For each $v \in \adom(A_{k+1})$, $S(v) = \{h(v)\}$.
            \item \textbf{Propagation:} For $i = k$ down to $1$, for each $u \in \adom(A_i)$, we compute $S(u)$ by merging the summaries of its neighbors in $R_i$:
            $S(u) = \textsf{bottom-k} \left( \bigcup_{(u, v) \in R_i} S(v) \right)$
            where $\textsf{bottom-k}$ retains the $K$ smallest distinct values.
            \item \textbf{Estimation:} For each $u \in \adom(A_1)$, we estimate $\deg(u)$. If $|S(u)| < K$, then $S(u)$ contains all reachable hash values, so $\widehat{\deg}(u) = |S(u)| = \deg(u)$. Otherwise, let $v_{K}$ be the $K$-th smallest value in $S(u)$ (i.e., $\max(S(u))$), the  estimator is $\widehat{\deg}(u) = \frac{K}{v_{K}}$.
        \end{itemize}
        
        To achieve the $(\epsilon, \delta)$-guarantee, we repeat this process $O(\log (N/\delta))$ times. For each $u$, we take the median of its estimates to boost the success probability, and finally sum these medians.
        
        \begin{algorithm}[t]
            \caption{$\textsc{ApproxCountChain}_{\epsilon,\delta}(\R)$}
            \label{alg:approx-count-chain}
            \KwIn{An instance $\R$ of $\chain$, approximation quality $\epsilon$ and error $\delta$.}
            \KwOut{An estimate of $|\chain(\R)|$.}
            $K \gets \lceil 8/\epsilon^2 \rceil$, $m \gets \lceil 12 \log(N/\delta) \rceil$\;
            Initialize map $L$ where $L[u] \gets []$ for all $u \in \adom(A_1)$\;
            \For{$j \in [m]$}{
                Generate a pairwise independent hash function $h: \dom(A_{k+1}) \to [0,1]$\;
                Initialize map $S$ where $S[v] \gets \{h(v)\}$ for all $v \in \adom(A_{k+1})$\;
                \For{$i \gets k$ \KwTo $1$}{
                    Initialize map $S'$\;
                    \lForEach{$(u, v) \in R_i$}{$S'[u] \gets \textsf{bottom-k}(S'[u] \cup S[v])$}
                    $S \gets S'$}
                \For{$u \in \adom(A_1)$}{
                    \lIf{$|S[u]| < K$}{
                        Append $|S[u]|$ to $L[u]$}
                    \lElse{
                        Append $\frac{K}{\max(S[u])}$ to $L[u]$}
                }
            }
            $s \gets 0$\;
            \lFor{$u \in \adom(A_1)$}{
                $s \gets s + \text{Median}(L[u])$}
            \Return $s$\;
        \end{algorithm}
        
        \paragraph{Analysis}
        The merging of KMV summaries at each step corresponds to the set union of reachable values. Since the size of each summary is bounded by $K = O(1/\epsilon^2)$, the merge operation for a tuple $(u,v)$ takes $O(K)$ time. Across all relations, we process each tuple exactly once. Thus, one pass of the algorithm takes $O(N \cdot K)$ time.
        
        Regarding the correctness, for each $u \in \adom(A_1)$, the KMV estimator $\widehat{\deg}(u)$ is an $\epsilon$-approximation of $\deg(u)$ with at least constant probability. By repeating the process $m = O(\log(N/\delta))$ times and taking the median, we boost the success probability for this specific $u$ to $1 - \delta/N$. Since $|\adom(A_1)| \le N$, by applying the union bound over all $u \in \adom(A_1)$, we guarantee that with probability at least $1-\delta$, the estimates for all $u$ are simultaneously $\epsilon$-approximations. Consequently, their sum is an $\epsilon$-approximation of $\OUT$. The total time complexity is $O(N \cdot \frac{1}{\epsilon^2} \log \frac{N}{\delta})$.
        
        \begin{theorem}\label{thm:chain-count-ub}
            For $\chain$ and an arbitrary instance $\R$, Algorithm~\ref{alg:approx-count-chain} returns an $\epsilon$-approximation of $|\chain(\R)|$ with probability at least $1-\delta$ in $\O(N)$ time.
        \end{theorem}


    \subsection{Uniform Sampling over Chain Query}
        \label{sec:sample-chain}
        
        We next show how to use the results from approximate counting to perform uniform sampling for $\chain$. 
        Recall that for each $u \in \adom(A_1)$, we obtain an estimator $\widehat{\deg}(u)$ such that $(1-\epsilon)\deg(u) \le \widehat{\deg}(u) \le (1+\epsilon)\deg(u)$ holds with high probability. To facilitate rejection sampling, we define an upper bound proxy for the degree: $\widetilde{\deg}(u) = \frac{\widehat{\deg}(u)}{1-\epsilon}.$
        It follows that $\deg(u) \le \widetilde{\deg}(u) \le \frac{1+\epsilon}{1-\epsilon}\deg(u)$. We use these proxy degrees to guide the sampling process towards ``heavy'' starting nodes, and then correct the bias using rejection sampling.

        \paragraph{Algorithm}
        The sampling procedure consists of a preprocessing phase and a sampling phase. In the preprocessing phase, we run the approximate counting algorithm to compute $\widetilde{\deg}(u)$ for all $u \in \adom(A_1)$. We then construct a data structure (e.g., an alias table or a prefix sum array) to support weighted sampling of $u \in \adom(A_1)$ proportional to $\widetilde{\deg}(u)$.
        
        In the sampling phase, we first sample a start node $u^*$ according to the weights $\widetilde{\deg}$. Then, we compute the exact set of reachable values $V_{u^*} = \{v \in \adom(A_{k+1}) \mid u^* \rightsquigarrow v\}$ by traversing the join graph starting from $u^*$. Note that $|V_{u^*}| = \deg(u^*)$. We draw a value $v^*$ uniformly at random from $V_{u^*}$. Finally, we accept the pair $(u^*, v^*)$ as a valid sample with probability $\frac{\deg(u^*)}{\widetilde{\deg}(u^*)}$. If rejected, we restart the sampling phase.
        
        \begin{algorithm}[t]
            \caption{$\textsc{SampleChain}(\R, \epsilon, \delta)$}
            \label{alg:sample-chain}
            \KwIn{An instance $\R$ of $\chain$, parameters $\epsilon, \delta$.}
            \KwOut{A uniform random sample from $\chain(\R)$.}
            \tcp{Suppose $\textsc{ApproxCountChain}_{\epsilon,\delta}(\R)$ returns $\widehat{\deg}(u)$ for all $u \in \adom(A_1)$. We then 
            compute $\widetilde{\deg}(u) = \frac{\widehat{\deg}(u)}{1-\epsilon}$ for all $u \in \adom(A_1)$. Let
            $W = \sum_{u \in \adom(A_1)} \widetilde{\deg}(u)$.}
            Sample $u^* \in \adom(A_1)$ with probability $\frac{\widetilde{\deg}(u^*)}{W}$\;
            Compute $V_{u^*} \gets \{v \in \adom(A_{k+1}) \mid u^* \rightsquigarrow v\}$ via graph traversal\;
            $\deg(u^*) \gets |V_{u^*}|$\;
            \lIf{$\deg(u^*) > \widetilde{\deg}(u^*)$}{\Return $\perp$}
            Sample $v^*$ uniformly from $V_{u^*}$\;
            \lIf{$\unif{0}{1} \le \frac{\deg(u^*)}{\widetilde{\deg}(u^*)}$}{\Return $(u^*, v^*)$}
        \end{algorithm}
        
        \paragraph{Analysis} The probability of returning the pair $(u, v) \in \chain(\R)$ in one iteration is the product of the probability of choosing $u$, the probability of choosing $v$ given $u$, and the acceptance probability:
        \begin{align*}
            \Pr[\text{return } (u,v)] &= \underbrace{\frac{\widetilde{\deg}(u)}{W}}_{\text{choose } u} \cdot \underbrace{\frac{1}{\deg(u)}}_{\text{choose } v} \cdot \underbrace{\frac{\deg(u)}{\widetilde{\deg}(u)}}_{\text{accept}} = \frac{1}{W}.
        \end{align*}
        Since this probability is independent of the specific pair $(u,v)$, the returned sample is uniform.
        
        Regarding the running time, the preprocessing takes $\O(N)$ time as established in Section~\ref{sec:count-chain}. In the sampling loop, computing $V_{u^*}$ takes $O(N)$ time in the worst case (traversing the join path). The expected number of iterations is determined by the acceptance probability. Since $\widetilde{\deg}(u) \le \frac{1+\epsilon}{1-\epsilon}\deg(u)$, the acceptance probability is at least $\frac{1-\epsilon}{1+\epsilon} \approx 1 - 2\epsilon$. Thus, the expected number of trials is $O(1)$. The total expected sampling time is $\O(N)$.

           \begin{theorem}\label{thm:chain-sample-ub}
            For $\chain$, there is an algorithm that can take as input an arbitrary instance $\R$ of input size $N$ and builds an index in $\O(N)$ time such that with probability at least $1-1/N^{O(1)}$, a uniform sample from the query result $\Q(\R)$ can be returned in $\O(N)$ expected time. 
        \end{theorem}

    \paragraph{Remark} It should be noted that this sampling algorithm only guarantees the uniformity with high probability, while our sampling algorithms for matrix and star queries can guarantee the uniformity for sure.

     \subsection{Lower Bounds for Approximate Counting Chain Query}
    \label{sec:lb-chain}
   
    So far, we have established matching upper and lower bounds for approximately counting and uniform sampling matrix queries, which correspond to the $2$-chain query. However, the problem becomes significantly harder when its length increases. In this section, we establish a lower bound $\Omega(\min\{N,\frac{N^2}{\OUT}\})$ for the $3$-chain query and $\Omega(N)$ for the $4$-chain query, which can also be generalized to other queries that contain this small join as a core. 

    \begin{theorem}
        \label{thm:count-chain-lb-3}
        For $\chain= \pi_{A,D}R_1(A,B) \Join R_2(B,C)\Join R_3(C,D)$, given any parameter $N, \OUT\in \mathbb{Z}^+$, any randomized algorithm that takes as input an arbitrary instance $\R$ of input size $N$ and output size $\OUT$, and returns a constant-approximation of $|\chain(\R)|$ with constant probability, requires at least $\Omega(\min\{N,\frac{N^2}{\OUT}\})$ time, under the property testing model.
    \end{theorem}
    
    \begin{proof}
        Consider an arbitrary constant $\theta \ge 2$ and two arbitrary integers $\Delta, L \in \mathbb{Z}^+$ such that $\Delta\cdot L \le \frac{N^2}{\theta}$ and $(\theta+1)L \le N$. We define two distributions of instances, $\D_0$ and $\D_1$. Let $\dom(A) = \{a_{ij}: i\in[3N], j\in [\sqrt{\Delta}]\}$, $ \dom(B) =\{b_i: i \in [3N]\}$, $\dom(C) =\{c_i: i \in [3N]\}$, and $\dom(D) = \{d_{ij}: i\in[3N], j\in [\sqrt{\Delta}]\}$. 
        
        First, we choose a subset $B_\alpha \subset \dom(B)$ of size $\frac{N}{\sqrt{\Delta}}$ and a subset $C_\alpha \subset \dom(C)$ of size $\frac{N}{\sqrt{\Delta}}$. Let $B_\beta = \dom(B) - B_\alpha$ and $C_\beta = \dom(C) - C_\alpha$.
        We fix $R_1 = \{(a_{ij},b_i): b_i \in B_\alpha, j\in [\sqrt{\Delta}]\}$ and $R_3 = \{(c_i,d_{ij}): c_i \in C_\alpha, j \in [\sqrt{\Delta}]\}$. Note that $|R_1| = |R_3| = |B_\alpha|\sqrt{\Delta} = N$.
        
        The two distributions differ only in the construction of $R_2$:
        \begin{itemize}[leftmargin=*]
        \item \textbf{Distribution $\D_0$:} We construct $R_2$ by adding $L$ random tuples from $B_\alpha \times C_\alpha$, $N-L$ random tuples from $B_\alpha \times C_\beta$, $N-L$ random tuples from $B_\beta \times C_\alpha$, and $L$ random tuples from $B_\beta \times C_\beta$. 
        The output size is $\OUT = L \cdot (\sqrt{\Delta})^2 = L \cdot \Delta$.
        
        \item \textbf{Distribution $\D_1$:} We construct $R_2$ by adding $\theta L$ random tuples from $B_\alpha \times C_\alpha$, $N-\theta L$ random tuples from $B_\alpha \times C_\beta$, $N-\theta L$ random tuples from $B_\beta \times C_\alpha$, and $\theta L$ random tuples from $B_\beta \times C_\beta$. 
        The output size is $\OUT_1 = \theta L \cdot \Delta$.
        \end{itemize}
        
        In both cases, $|R_2| = 2N$. The constraints on $\Delta$ and $L$ ensure that the subsets $B_\alpha \times C_\alpha$, etc., are large enough to support these selections without replacement.
        
        Let $\mathcal{A}$ be an algorithm that distinguishes $\D_0$ from $\D_1$. Since $\OUT_1 = \theta \cdot \OUT$ and $\theta \ge 2$, any constant-factor approximation algorithm must distinguish these two cases. The difficulty lies in identifying tuples in the ``critical region'' $B_\alpha \times C_\alpha$.
        Although the distributions also differ in regions involving $B_\beta$ or $C_\beta$, distinguishing based on those regions is statistically harder. The regions $B_\alpha \times C_\beta$ and $B_\beta \times C_\alpha$ contain $\Theta(N)$ tuples with a count difference of only $\Theta(L)$, requiring $\Omega((N/L)^2)$ queries to distinguish the bias. The region $B_\beta \times C_\beta$ has a domain size of $\Theta(N^2)$ with only $\Theta(L)$ tuples, requiring $\Omega(N^2/L)$ queries to find a witness. Thus, the most efficient strategy for $\mathcal{A}$ is to focus on $B_\alpha \times C_\alpha$. Consider the queries made to $R_2$:
        \begin{itemize}[leftmargin=*]
            \item \textbf{test tuple:} $\mathcal{A}$ can check if $(b,c) \in R_2$. To distinguish the distributions, $\mathcal{A}$ must find tuples in $B_\alpha \times C_\alpha$. The size of this domain is $|B_\alpha||C_\alpha| = \frac{N^2}{\Delta}$. In $\D_0$, the density of tuples is $\rho_0 = \frac{L}{N^2/\Delta}$, and in $\D_1$, it is $\rho_1 = \frac{\theta L}{N^2/\Delta}$. To distinguish these densities, $\mathcal{A}$ requires $\Omega(\frac{1}{\rho_1}) = \Omega(\frac{N^2}{\theta L \Delta})$ queries.
            
            \item \textbf{sample tuple:} $\mathcal{A}$ can sample uniformly from $R_2$. The probability of sampling a tuple from $B_\alpha \times C_\alpha$ is $\frac{L}{2N}$ in $\D_0$ and $\frac{\theta L}{2N}$ in $\D_1$. To distinguish these biases, $\mathcal{A}$ requires $\Omega(\frac{2N}{\theta L}) = \Omega(\frac{N}{L})$ samples.
            
            \item \textbf{compute degree:} In both distributions, the total number of edges incident to $B_\alpha$ in $R_2$ is exactly $N$. Since the endpoints are chosen uniformly at random, the marginal distribution of the degree of any specific node $b \in B_\alpha$ is identical in both $\D_0$ and $\D_1$. Thus, the degree oracle provides no information to distinguish the distributions.
    
            \item \textbf{access neighbor:} $\mathcal{A}$ can query the $j$-th neighbor of a value $b \in B_\alpha$. Across all values in $B_\alpha$, there are exactly $N$ incident edges, creating a total of $N$ valid query slots (pairs of $(b, j)$). In our construction, the $L$ (in $\D_0$) or $\theta L$ (in $\D_1$) critical edges connecting $B_\alpha$ to $C_\alpha$ are assigned to these slots uniformly at random (effectively by randomizing the neighbor orderings). Consequently, querying a specific slot $(b, j)$ is statistically equivalent to sampling an edge uniformly at random from the $N$ edges incident to $B_\alpha$. The probability that such a query reveals a critical edge is $\frac{L}{N}$ in $\D_0$ and $\frac{\theta L}{N}$ in $\D_1$. As with the \textbf{sample tuple} oracle, distinguishing these cases requires $\Omega(\frac{N}{L})$ queries.
        \end{itemize}
        
        Combining these above, any algorithm requires $\Omega\left(\min\left\{\frac{N}{L}, \frac{N^2}{L \Delta}\right\}\right)$ time to distinguish $\mathcal{D}_0$ from $\mathcal{D}_1$. The lower bound $\Omega\left(\min\{N, \frac{N^2}{\OUT}\}\right)$ follows by choosing $\Delta=\Theta(\OUT)$ and $L=O(1)$.
    \end{proof}

       \begin{theorem}
        \label{thm:count-chain-lb-4}
        For $\chain= \pi_{A,E}R_1(A,B) \Join R_2(B,C)\Join R_3(C,D) \Join R_4(D,E)$, given any parameter $N, \OUT\in \mathbb{Z}^+$, any randomized algorithm that takes as input an arbitrary instance $\R$ of input size $N$ and output size $\OUT$, and returns a constant-approximation of $|\chain(\R)|$ with constant probability, requires at least $\Omega(N)$ time, under the property testing model.
    \end{theorem}

          \begin{proof}
            If $\OUT \le N$, the lower bound follows from the 3-chain case (Theorem~\ref{thm:count-chain-lb-3}) by setting $R_1$ as a one-to-one mapping. We focus on the case where $\OUT > N$.
    
            Consider any constant $\theta \ge 2$. We choose parameters $L = \Theta(1)$ and $\Delta = \Theta(\OUT)$ such that  $\Delta \cdot L \le \frac{N^2}{\theta}$ and $(\theta+1)L \le N$ hold. We define two distributions of instances, $\D_0$ and $\D_1$. Let $\dom(A) = \{a_{ij}: i\in[3N], j\in [\sqrt{\Delta}]\}$, $\dom(B) =\{b_i: i \in [3N]\}$, $\dom(C) = \{c_i: i\in [3N]\}$, $\dom(D) =\{d_i: i \in [3N]\}$, and $\dom(E) = \{e_{ij}: i\in[3N], j\in [\sqrt{\Delta}]\}$. We fix subsets $B_\alpha \subset \dom(B)$ and $D_\alpha \subset \dom(D)$, each of size $\frac{N}{\sqrt{\Delta}}$. Let $B_\beta = \dom(B) - B_\alpha$ and $D_\beta = \dom(D) - D_\alpha$. We also choose a subset $C_\alpha \subset \dom(C)$ of size $N$ and partition it into $\ell = \frac{N^2}{\OUT}$ groups $C_1, \dots, C_\ell$, each of size $\frac{\OUT}{N}$. We fix $R_1 = \{(a_{ij},b_i): b_i \in B_\alpha, j\in [\sqrt{\Delta}]\}$ and $R_4 = \{(d_i,e_{ij}): d_i \in D_\alpha, j \in [\sqrt{\Delta}]\}$. Note that $|R_1| = |R_4| = N$. The distributions differ in $R_2$ and $R_3$:
            \begin{itemize}[leftmargin=*]
             \item \textbf{Distribution $\D_0$:} We select $L$ distinct pairs from $B_\alpha \times D_\alpha$ uniformly at random. For each selected pair $(b_i,d_j)$, we exclusively assign a distinct group $C_h$ from the partition of $C_\alpha$, pick a value $c_m$ uniformly at random from $C_h$, and add $(b_i,c_m)$ to $R_2$ and $(c_m, d_j)$ to $R_3$. 
             To mask these critical tuples, we add random noise tuples (e.g., from $B_\alpha \times C_\beta$, etc.) until $|R_2|=|R_3|=2N$, similar to the construction in Theorem~\ref{thm:count-chain-lb-3}.
             The output size is $\OUT_0 = L \cdot (\sqrt{\Delta})^2 = L \cdot \Delta$.
           
            \item \textbf{Distribution $\D_1$:} The construction is identical, except we select $\theta L$ pairs instead of $L$. The output size is $\OUT_1 = \theta L \cdot \Delta$.
            \end{itemize}
            
            Since $\OUT_1 = \theta \cdot \OUT_0$, any constant-factor approximation must distinguish $\D_0$ from $\D_1$. Consider the queries made:
        \begin{itemize}[leftmargin=*]
            \item \textbf{test tuple:} To distinguish the distributions, the algorithm must identify the ``critical'' paths $b \to c \to d$. There are $|B_\alpha|\cdot|D_\alpha| = \frac{N^2}{\Delta}$ candidate pairs in $B_\alpha \times D_\alpha$. The probability that a random pair is critical is $\rho = \frac{L}{N^2/\Delta}$ (in $\D_0$). 
            Even if the algorithm successfully guesses a critical pair $(b,d)$ (which requires $\Omega(1/\rho)$ trials), and even if it knows the specific group $C_h$ assigned to this pair, it must still find the specific witness $c_m \in C_h$. Since $|C_h| = \frac{\OUT}{N}$, finding the witness requires $\Omega(\frac{\OUT}{N})$ queries.
            The total complexity is:
            $\Omega\left(\frac{1}{\rho} \cdot \frac{\OUT}{N}\right) = \Omega\left(\frac{N^2}{L \Delta} \cdot \frac{\OUT}{N}\right).$
            Substituting $\OUT = \Theta(L \Delta)$, the cost is $\Omega(N)$.
    
            \item \textbf{sample tuple:} The critical tuples in $R_2$ (those in $B_\alpha \times C_\alpha$) constitute a fraction of $\frac{L}{2N}$ of the relation size. Distinguishing the bias between $L$ and $\theta L$ requires $\Omega(\frac{N}{L})$ samples. Since we chose $L = \Theta(1)$, this cost is $\Omega(N)$.
            \item \textbf{compute degree or access neighbor:} As in Theorem~\ref{thm:count-chain-lb-3}, the marginal degree distributions of $B_\alpha$ and $D_\alpha$ are masked by the noise tuples. Accessing a neighbor is statistically equivalent to sampling, yielding the same $\Omega(N)$ lower bound. \hfill \qedsymbol\end{itemize}\renewcommand{\qedsymbol}{}
        \end{proof}

     \subsection{Lower Bounds for Uniform Sampling over Chain Query}
     \label{sec:lb-chain-sample}

    We establish the lower bound for uniform sampling chain queries via the reduction from approximate counting. The reduction in Proposition~\ref{prop:reduction} requires that the sampling algorithm returns a query result with probability $\frac{\OUT}{W}$ for some parameter $W$, so the lower bounds derived from this reduction apply to this class of sampling algorithms. We call such algorithms \emph{$W$-uniform sampling} algorithms: given an index built during preprocessing, each invocation returns a uniform sample from $\Q(\R)$ with probability $\frac{\OUT}{W}$, and always returns the value of $W$. Note that all sampling algorithms presented in this paper fall into this class, as do all prior sampling algorithms for join-project queries that we are aware of.

    \begin{theorem}
    \label{thm:chain-sample-lb-3}
        For $\chain$ with $k =3$, given parameters $N, \OUT$ 
        , no $W$-uniform sampling algorithm can preprocess an instance $\R$ of input size $N$ and output size $\OUT$ in
        $O(\min\{N^{1-\gamma}, (\frac{N^{2}}{\OUT})^{1-\gamma}\})$ time and generate a sample in $O(\min\{N^{1-\gamma}, (\frac{N^{2}}{\OUT})^{1-\gamma}\})$ time, for any arbitrary small constant $\gamma > 0$.
    \end{theorem}

    \begin{proof}
        Following Proposition~\ref{prop:reduction}, Theorem~\ref{thm:count-chain-lb-3} implies a lower bound on $W$-uniform sampling over chain queries. Specifically, if there exists a $W$-uniform sampling algorithm with preprocessing time $O(\min\{N^{1-\gamma}, (\frac{N^{2}}{\OUT})^{1-\gamma}\})$ that can generate a sample in $O(\min\{N^{1-\gamma}, (\frac{N^{2}}{\OUT})^{1-\gamma}\})$, we construct an approximate counting algorithm by invoking the sampling procedure $\O(1)$ times. The total runtime would be $O(\min\{N^{1-\gamma}, (\frac{N^{2}}{\OUT})^{1-\gamma}\})$, which contradicts the lower bound in Theorem~\ref{thm:count-chain-lb-3}.
    \end{proof}


    \begin{theorem}
    \label{thm:chain-sample-lb-4}
        For $\chain$ with $k \ge 4$, no $W$-uniform sampling algorithm can preprocess an instance $\R$ of input size $N$ in $O(N^{1-\gamma})$ time and generate a sample from $\chain(\R)$ in $O(N^{1-\gamma})$ time, for any arbitrary small constant $\gamma > 0$.
    \end{theorem}

    \begin{proof}
        Following Proposition~\ref{prop:reduction}, Theorem~\ref{thm:count-chain-lb-4} automatically implies a lower bound on $W$-uniform sampling over chain queries. Specifically, if there exists a $W$-uniform sampling algorithm with preprocessing time $O(N^{1-\gamma})$ that can generate a sample in $O(N^{1-\gamma})$ time, we can construct an approximate counting algorithm by invoking the sampling procedure $\O(1)$ times. The total running time would be $\O(N^{1-\gamma})$, which contradicts the $\Omega(N)$ lower bound established in Theorem~\ref{thm:count-chain-lb-4}.
    \end{proof}

    Note that, unlike the combinatorial lower bounds for matrix and star queries, this result relies on the information-theoretic hardness of distinguishing distributions (as used in Theorem~\ref{thm:count-chain-lb-3} and Theorem~\ref{thm:count-chain-lb-4}) and thus applies to any $W$-uniform sampling algorithm, not just combinatorial ones.

       \section{Conclusion}
    \label{sec:conclusion}
    This work initiates a rigorous study into the fine-grained complexity of uniform sampling and approximate counting for join-project queries. We established the first asymptotically optimal algorithms for fundamental query classes -- specifically matrix, star, and chain queries -- marking a pivotal step towards output-optimality for general join-project queries. Our findings open several avenues for future research:
    \begin{itemize}[leftmargin=*]
    \item {\em A Unified Parameterized Framework:} A primary open problem is to synthesize the techniques developed for star and chain queries into a unified algorithm capable of handling arbitrary join-project queries. Note that the techniques introduced for star queries can be easily generalized to all join-project queries (see Appendix \ref{appendix:generalization}), but they are already not optimal for even for the length-4 chain query. Furthermore, a rigorous parameterized complexity analysis is required to characterize algorithmic performance relative to both input and output sizes. 

    \item {\em Support Selection Predicate:} So far, our algorithms have investigated only join and projection operators. A natural extension is to incorporate selection predicates to fully support general conjunctive queries and to determine the impact of selectivity on sampling complexity. Very recently, \cite{huang2026acyclic} investigated the sampling indices for acyclic joins when the point selection predicates exist. But it is still open for more complicated predicates.
    
    \item {\em Beyond Combinatorial Limits:} While our results define the boundaries of combinatorial approaches in some scenarios, the application of algebraic techniques remains largely unexplored. A compelling frontier is to determine whether methods such as fast matrix multiplication can circumvent existing combinatorial barriers. Recent efforts have shown that fast matrix multiplication can be used to accelerate counting subgraph patterns \cite{alon1995color,DBLP:journals/ipl/KloksKM00, tvetek2022approximate, censor2024fast,censor2025output,eisenbrand2004complexity, DBLP:conf/stoc/DalirrooyfardMW24,DBLP:conf/stoc/CurticapeanDM17, DBLP:journals/siamdm/KowalukLL13}, , such as triangles, cycles, and cliques. But it's still unknown how much these techniques can be extended to general join queries or even join-project queries. 
\end{itemize}
    \section*{Acknowledgement}
       Xiao Hu was supported by the Natural Sciences and Engineering Research Council of Canada (NSERC) Discovery Grant. Jinchao Huang was supported by the Overseas Research Attachment Program of The Chinese University of Hong Kong during his research visit to the University of Waterloo. We thank Sepehr Assadi for invaluable discussions and feedback on this problem, as well as the anonymous reviewers for their constructive comments. Jinchao Huang also expresses gratitude to Sibo Wang for supporting his overseas research at the University of Waterloo.
       
    \bibliographystyle{abbrvnat}
    \bibliography{paper}

    \received{December 2025} \received[accepted]{February 2026}

    \appendix
    \section{Extension to General Join-Project Queries}
    \label{appendix:generalization}

    \begin{algorithm}[t]
        \caption{\textsc{SampleJoinProject}$(\R)$}
        \label{alg:sample-join-project}
        \KwIn{An instance $\R$ for the join-project query $\Q=(\V,\E,\y)$, with some auxiliary indices built for $\R$ if needed.}
        \KwOut{A uniform sample of the query result $\Q(\R)$.}
        \While{\true}{
        $s\gets $\textsc{JoinSample}$(\Q,\R)$ \cite{zhao2018random,deng2023join, kim2023guaranteeing}\;
        \lIf{$s \neq \perp$}{\textbf{break}}
        }
        \lIf{\textup{\textsc{AcceptJoinProject}}$(\Q,\R,s) = \true$}{\Return $\pi_{\y} s$; \qquad \qquad \qquad $\blacktriangleright$ \textup{Algorithm~\ref{alg:accept-join-project}}}
        \lElse{\Return $\perp$}
    \end{algorithm}

    Our framework introduced for matrix and star queries can be easily extended to arbitrary join-project queries. We first recall the AGM bound \cite{atserias2008size}. For a join query $q=(\V,\E)$, let $\mathbf{N}: \E \to \mathbb{Z}^+$ be the size function, and $w$ be a fractional edge covering such that $\sum_{e: x \in e} w(e)\ge 1$ for each attribute $A \in \V$. The AGM bound \cite{atserias2008size} states 
    \begin{equation}
        \label{eq:AGM}
        \max_{\R:|R_e| \le \mathbf{N}_e, \forall e\in \E} |\Q_\Join(\R)|\le \prod_{e \in \E} \mathbf{N}_e^{w(e)}
    \end{equation}
    We also use $\textsf{AGM}(q, \mathbf{N})$ to denote the smallest upper bound for $q$ under the input function $\mathbf{N}$, i.e., the smallest right-hand-side of E.q. (\ref{eq:AGM}) over all possible fractional edge coverings.
    As a special case, when $\mathbf{N}_e = N$ for each $e \in \E$, we have $\textsf{AGM}(q, \mathbf{N}) = N^{\sum_{e \in \E} w(e)}$. The tightest upper bound would be the one minimizing $\sum_{e \in \E} w(e)$, i.e., $\rho^*(q)$. For a join-project query $\Q=(\V,\E,\y)$, let $\Q_{\Join} = (\V,\E)$ be the underlying full join query of $\Q$. Given an instance $\R$ for $\Q$, let $\OUT_\Join = |\Q_\Join(\R)|$ be the number of full join results. From the AGM bound, $\OUT_\Join \le N^{\rho^*(\Q_{\Join})}$.
    
    \paragraph{Algorithm} As described in Algorithm \ref{alg:sample-join-project}, we first get a sample from the full join result using the existing sampling algorithm \cite{deng2023join, kim2023guaranteeing}, say $s$ and then accept it with probability $\frac{1}{\deg(\pi_\y s)}$, where $\deg(\pi_\y s)$ is the number of full join results that $\pi_\y s$ appears in. Following exactly the same analysis, we can show that each query result $t \in \Q(\R)$ is sampled with probability $$\sum_{s \in \Q_{\Join}(\R): \pi_{\y} s = t} \frac{1}{|\Q_{\Join}(\R)|} \cdot \frac{1}{\deg(t)} = \frac{1}{|\Q_{\Join}(\R)|},$$
    it is therefore a uniform sample of the query result. 

    To implement this algorithm, we will need a primitive {\sc JoinSample} proposed by prior work \cite{zhao2018random,deng2023join, kim2023guaranteeing}, which takes as input a full join query $\Q_\Join$ and an instance $\R$, and always returns a uniform sample from the full join result $\Q_\Join(\R)$ with probability $1$ for acyclic joins and $\frac{|\Q_\Join(\R)|}{\textrm{\upshape AGM}(\Q_\Join, \{|R_e|: e \in \E\})
    }$ for cyclic joins. We will repeat this primitive until a valid join result is returned.
    
    Below, we focus on the {\sc AcceptJoinProject} primitive. Suppose $s$ is a full join result returned by the sampling procedure {\sc JoinSample}. Let $\Q_{\bar{\y}}$ be the full join query induced by non-output attributes. 
    Let $\R_s$ be the induced sub-instance of $\R$ that can be joined with $s$ on output attributes. Note that any join result of $\Q_{\bar{\y}}(\R_s)$ will form with $t$ as a full join result of $\Q_\Join(\R)$. 
    The following steps follow Section \ref{sec:sample-matrix}. As computing $\deg(t)$ is too expensive, we simulate this probability using Proposition \ref{prop:simulation}. 
    
    To apply this proposition, we first identify $S$ as the sampling universe of $\Q_{\bar{\y}}$ on $\R_s$ defined by the {\sc JoinSample} primitive, where $|S| = \textrm{AGM}(\Q_{\bar{\y}}, \{|R_e|: e \in \E\})$. Note that $\pi_{\V-\y} s$ is a valid join result of $\Q_{\bar{\y}}(\R_s)$ (a ``success'') because $s$ was sampled from the full join result $\Q_\Join(\R)$. We construct the random variable $Y$ by sampling from the \emph{other} elements in $S$. Specifically, we sample without replacement from $S - \{\pi_{\V-\y} s\}$ and count the number of ``failures'' ($F$) encountered before finding another valid join result in $\Q_{\bar{\y}}(\R_s)$. If no other valid join value exists (i.e., $\pi_{\V-\y} s$ is unique), we continue until $S - \{\pi_{\V-\y} s\}$ is exhausted. We define $Y = F+1$. It iterates through random samples from $S$ by invoking the {\sc JoinSample} primitive. The expectation of this specific construction is $\expt[Y] = \frac{|S|}{\deg(t)}$. Finally, we generate a random integer $X$ uniformly from $[|S|]$. If $X \le F+1$, we return \true{}; otherwise, we return \false{}. By Proposition~\ref{prop:simulation}, the acceptance probability is exactly $\frac{1}{|S|} \cdot \frac{|S|}{\deg(t)} = \frac{1}{\deg(t)}$. Additionally, if the residual query $\Q_{\bar{\y}}$ is acyclic, the sampling universe can be easily computed since every acyclic join has an optimal fractional edge covering that is also integral \cite{hu2021cover}. Hence, the Cartesian product of those selected relations (i.e., assigned weight 1 in the optimal fractional edge covering) forms the sampling universe. Moreover, no additional auxiliary indices are required, as relations in $\R_s$ are well indexed.

    \begin{algorithm}[t]
        \caption{\textsc{AcceptJoinProject-I}$(\Q, \R, s)$}
        \label{alg:accept-join-project}
        \KwIn{A join-project query $\Q$, an instance $\R$ for the join-project query $\Q=(\V,\E,\y)$, and a full join result $s \in \Q_\Join(\R)$ with $t = \pi_{\y} s$.}
        \KwOut{A Boolean value indicating acceptance.}
        $\Q_{\bar{\y}} \gets (\V-\y, \{e -\y: e\in \E\},\V-\y)$\;
        $\R_s \gets \{\pi_{e - \y} \left(R_e \ltimes (\pi_\y s)\right): e \in \E\}$\;
        $F \gets 0$\;
        \tcp{For acyclic queries, line \ref{step:setup} can be skipped, and line \ref{sample-universe} is the Cartesian product of a subset of relations such that every attribute of $\V-\y$ appears in at least one selected relation}
        Initialize the auxiliary indices for {\sc JoinSample} if needed\;\label{step:setup}
        $S \gets $ the sampling universe defined for $\Q_{\bar{\y}}$ on $\R_s$  \cite{deng2023join}\;\label{sample-universe}
        \tcp{We sample from $S$ using uniform sampling without replacement}
        \While{at least one element in $S$ is not visited yet}{
            $s' \leftarrow$ a random sample from non-visited elements in $S$ using $\textsc{JoinSample}(\Q_{\bar{\y}}, \R_s)$\;
        \lIf{$s' = \pi_{\V -\y} s$}{\textbf{continue}}
        \lIf{$s' \notin \Q_{\bar{\y}}(\R_s)$}{$F \gets F + 1$}
        \lElse{\textbf{break}}
        }
        \lIf{$\randint(1, \textrm{\upshape AGM}(\Q_{\bar{\y}}, \{|R_e \ltimes (\pi_\y s)|: e\in \E\})) \le F+1$}{\Return \true}
        \Return \false\;
    \end{algorithm}

      \begin{algorithm}[t]
        \caption{\textsc{AcceptJoinProject-II}$(\Q, \R, s)$}
        \label{alg:accept-join-project-alternative}
        \KwIn{A join-project query $\Q$, an instance $\R$ for the join-project query $\Q=(\V,\E,\y)$, and a full join result $s \in \Q_\Join(\R)$ with $t = \pi_{\y} s$.}
        \KwOut{A Boolean value indicating acceptance.}
        $\Q_{\bar{\y}} \gets (\V-\y, \{e -\y: e\in \E\},\V-\y)$\;
        $\R_s \gets \{\pi_{e - \y} \left(R_e \ltimes (\pi_\y s)\right): e \in \E\}$\;
        Compute $|\Q_{\bar{\y}}(\R_s)|$\;
        \Return \true with probability $\frac{1}{|\Q_{\bar{\y}}(\R_s)|}$ and \false with probability $1-\frac{1}{|\Q_{\bar{\y}}(\R_s)|}$\;
    \end{algorithm}

    \paragraph{Analysis} 
        We now analyze the time complexity of Algorithm~\ref{alg:sample-join-project}. We also distinguish our analysis for acyclic and cyclic queries separately, since there could be much more efficient implementation of primitives for acyclic queries. 
        
        \underline{\emph{Analysis for acyclic queries.}} The preprocessing step takes $O(N)$ time. After preprocessing, each invocation of {\sc JoinSample} takes $\O(1)$ time and always returns a successful sample. Hence the while-loop at 
        Lines 1-3 
        takes $O(1)$ time. The invocation of Algorithm~\ref{alg:accept-join-project} at
        Line 4
        takes $O(N+F+1)$ time. From our analysis above, the expectation $\expt[F+1] = \frac{|S|}{\deg(s)} = \frac{\textrm{\upshape AGM}(\Q_{\bar{\y}}, \{|R_e \ltimes (\pi_\y s)|: e\in \E\})}{\deg(s)}$. Let $\Q_\y = (\V, \E \cup \{\y\})$ be the full join by adding an additional relation that contains all output attributes of $\y$. 
        So, following the decomposition lemma in \cite{ngo2014skew}:
        \begin{align*}
            \sum_{t\in \Q(\R)} \textrm{\upshape AGM}(\Q_{\bar{\y}}, \{|R_e \ltimes t|: e\in \E\}) \le N^{\sum_{e\in \E} w(e)} \cdot \OUT^{w(\y)}
        \end{align*}
        for any fractional edge covering $w$ of $\Q_\y$. Hence, Algorithm~\ref{alg:sample-join-project} takes (omitting the big-$O$):
        \begin{align*}
        \expt[\text{time per invocation}] &= 1 + \sum_{t\in \Q(\R)} \frac{\deg(t)}{\OUT_\Join} \cdot \frac{\textrm{\upshape AGM}(\Q_{\bar{\y}}, \{|R_e \ltimes t|: e\in \E\})}{\deg(t)}\\
        & \le 1 + \frac{\sum_{t\in \Q(\R)} \textrm{\upshape AGM}(\Q_{\bar{\y}}, \{|R_e \ltimes t|: e\in \E\})}{\OUT_\Join} \\
        & \le 
           1 + \frac{N^{\sum_{e\in \E} w(e)} \cdot \OUT^{w(\y)}}{\OUT_\Join}
        \end{align*}
        where $w$ is an arbitrary fractional edge covering of $\Q_\y$. 
        
        As an alternative, Algorithm \ref{alg:accept-join-project-alternative} only takes $O(N)$ time to compute the join size for acyclic queries.

        Algorithm~\ref{alg:sample-join-project} succeeds with probability $\frac{\OUT}{\OUT_\Join}$, so the expected number of invocations before getting one successful sample is $\frac{\OUT_\Join}{\OUT}$. The total expected cost of exploiting the power of two strategies is $O\left(\min\left\{
        \frac{\OUT_\Join}{\OUT}
        + N^{\sum_{e\in \E} w(e)} \cdot \OUT^{w(\y)-1}, \frac{N \cdot \OUT_\Join}{\OUT}\right\}\right)$.
        
        \underline{\emph{Analysis for cyclic queries.}} The preprocessing step takes $O(N)$ time. After preprocessing, each invocation of {\sc JoinSample} takes $\O(1)$ time. As {\sc JoinSample} returns a successful sample with probability $\frac{\OUT_\Join}{N^{\rho^*(\Q_\Join)}}$ for cyclic queries, hence the while-loop at 
        Lines 1-3 
        takes $O(\frac{N^{\rho^*(\Q_\Join)}}{\OUT_\Join})$ time for cyclic queries. The invocation of Algorithm~\ref{alg:accept-join-project} at
        Line 4
        takes $O(F+1)$ time. From our analysis above, the expectation $\expt[F+1] = \frac{|S|}{\deg(s)} = \frac{\textrm{\upshape AGM}(\Q_{\bar{\y}}, \{|R_e \ltimes (\pi_\y s)|: e\in \E\})}{\deg(s)}$. Hence, Algorithm~\ref{alg:sample-join-project} takes (omitting the big-$O$):
        \begin{align*}
        \expt[\text{time per invocation}] &= \frac{N^{\rho^*(\Q_\Join)}}{\OUT_\Join} + N + \sum_{t\in \Q(\R)} \frac{\deg(t)}{\OUT_\Join} \cdot \frac{\textrm{\upshape AGM}(\Q_{\bar{\y}}, \{|R_e \ltimes t|: e\in \E\})}{\deg(t)}\\
        & \le \frac{N^{\rho^*(\Q_\Join)}}{\OUT_\Join} + N + \frac{\sum_{t\in \Q(\R)} \textrm{\upshape AGM}(\Q_{\bar{\y}}, \{|R_e \ltimes t|: e\in \E\})}{\OUT_\Join}\\
        & \le \frac{N^{\rho^*(\Q_\Join)}}{\OUT_\Join} + N + \frac{N^{\sum_{e\in \E} w(e)} \cdot \OUT^{w(\y)}}{\OUT_\Join}
        \end{align*}
        expected time for cyclic queries. 
        Algorithm~\ref{alg:sample-join-project} succeeds with probability $\frac{\OUT}{\OUT_\Join}$, so the expected number of invocations before getting one successful sample is $\frac{\OUT_\Join}{\OUT}$. The total expected cost is:
        \begin{align*}
        O\left(\frac{N^{\rho^*(\Q_\Join)}}{\OUT}
        + \frac{N\cdot \OUT_\Join}{\OUT} + N^{\sum_{e\in \E} w(e)} \cdot \OUT^{w(\y)-1}\right)
        \end{align*}

        Putting everything together, we obtain:
        \begin{theorem}
        \label{the:general-acyclic-join-project}
        For an acyclic join-project query $\Q=(\V, \E,\y)$, there is an algorithm that can take as input an arbitrary instance $\R$ of input size $N$, output size $\OUT$ and full join size $\OUT_\Join$, and builds an index in $O(N)$ time such that a uniform sample from the query result $\Q(\R)$ can be returned in $$O\left(
       \min\left\{
        \frac{\OUT_\Join}{\OUT}
        + N^{\sum_{e\in \E} w(e)} \cdot \OUT^{w(\y)-1}, \ \frac{N \cdot \OUT_\Join}{\OUT}\right\}\right)$$ expected time, where $w$ is the fractional edge covering of $\Q_\y = (\V, \E \cup \{\y\})$.
    \end{theorem}

    \paragraph{Remark 1} For an acyclic join-project query $\Q = (\V,\E,\y)$, let $q= (\V,\E)$ be the underlying full join query. Theorem \ref{the:general-acyclic-join-project} can always improve the state-of-the-art method \cite{kim2023guaranteeing, deng2023join, chen2020random} since $\OUT_\Join \le N^{\rho^*(\Q_\Join)}$ and $\min_w N^{\sum_{e\in \E} w(e)} \cdot \OUT^{w(\y)} \le N^{\rho^*(\Q_\Join)}$, since $w$ degenerates to any fractional edge covering for $q$ together with assigning $w(\y) = 0$.  Consider an example query $\Q$ with $\V=\{A_1,A_2,B,C,D\}$, $\E=\{\{A_1,B\}, \{A_2,B\}, \{B,C\}, \{C,D\}\}$ and $\y= \{A_1,A_2,D\}$. On $\Q$, we observe that $\OUT_\Join \le N \cdot \OUT$ and $\min_w N^{\sum_{e\in \E} w(e)} \cdot \OUT^{w(\y)} \le N \cdot \OUT$. By our loose analysis, the state-of-the-art method would take $O(\frac{N^{3}}{\OUT})$ expected time to obtain a uniform sample while Algorithm \ref{alg:sample-join-project} takes at most $O(N)$ expected time to obtain a uniform sample.

    \paragraph{Remark 2} Theorem \ref{the:general-acyclic-join-project} can be extended to support approximate counting following the reduction in \cite{chen2020random}. More specifically, each trial of Algorithm \ref{alg:sample-join-project} succeeds in returning a valid query result with probability $\frac{\OUT}{\OUT_\Join}$ and $\OUT_\Join$ can be computed exactly in $O(N)$ for acyclic queries. Putting everything together, we obtain:

        \begin{theorem}
        \label{the:general-acyclic-join-project}
        For an acyclic join-project query $\Q=(\V, \E,\y)$, there is an algorithm that can take as input an arbitrary instance $\R$ of input size $N$, output size $\OUT$ and full join size $\OUT_\Join$, and returns a constant-approximation of $\OUT$ with constant probability in $$O\left(
        N + \frac{\OUT_\Join}{\OUT} + \min_{w} N^{\sum_{e\in \E} w(e)} \cdot \OUT^{w(\y)-1}\right)$$ time, where $w$ is the fractional edge covering of $\Q_\y = (\V, \E \cup \{\y\})$.
    \end{theorem}
    
    \begin{theorem}
        \label{the:general-cyclic-join-project}
        For a cyclic join-project query $\Q=(\V, \E,\y)$, there is an algorithm that can take as input an arbitrary instance $\R$ of input size $N$, output size $\OUT$ and full join size $\OUT_\Join$, and builds an index in $O(N)$ time such that a uniform sample from the query result $\Q(\R)$ can be returned in $$O\left(
        \frac{N^{\rho^*(\Q_\Join)}}{\OUT}
        + \frac{N\cdot \OUT_\Join}{\OUT} + \min_{w} \frac{N^{\sum_{e\in \E} w(e)}}{\OUT^{1-w(\y)}}\right)$$ expected time, where $\rho^*(\Q_\Join)$ is the fractional edge covering number of $\Q_\Join = (\V,\E)$, and $w$ is the fractional edge cover of $\Q_\y = (\V, \E \cup \{\y\})$.
    \end{theorem}
   
    \paragraph{Remark 3} For cyclic join-project queries, even if we could bypass this initialization step, the leading term $\frac{N^{\rho^*(\Q_\Join)}}{\OUT}$ in our total expected time is already as costly as the state-of-the-art methods \cite{kim2023guaranteeing, deng2023join, chen2020random}, hence Theorem \ref{the:general-cyclic-join-project} does not gain any improvement for cyclic queries.

   \end{document}

%% file: def/jc-def.tex
\usepackage[normalem]{ulem}
\newcommand\redsout{\bgroup\markoverwith{\textcolor{red}{\rule[0.5ex]{2pt}{1pt}}}\ULon}

\theoremstyle{plain}
\newtheorem{theorem}{Theorem}[section]
\newtheorem{proposition}[theorem]{Proposition}
\newtheorem{lemma}[theorem]{Lemma}

\theoremstyle{definition}
\newtheorem{definition}[theorem]{Definition}

\newtheorem{problem}[theorem]{Problem}

\newcommand{\unif}[2]{\textup{{\texttt{Uniform}}}(#1,#2)}

\SetKw{Continue}{continue}
\SetKw{Break}{break}